\documentclass[11pt,a4paper,reqno]{amsart}
\usepackage{amsmath,amssymb,amsthm,amsfonts}
\usepackage[export]{adjustbox}
\usepackage{amsbsy}
\usepackage[utf8]{inputenc}
\usepackage[normalem]{ulem}
\usepackage{xspace}
\usepackage{cancel}
\usepackage[greek,english]{babel}
\usepackage{url}
\usepackage{wrapfig}
\usepackage{enumitem}
\usepackage{multicol}
\usepackage{moreenum}
\usepackage{xcolor}
\usepackage{longtable}
\usepackage{array}

\usepackage{graphicx}

\usepackage[pagewise]{lineno} 
\newcounter{dummy}
\makeatletter
\newcommand\myitem[1][]{\item[#1]\refstepcounter{dummy}\def\@currentlabel{#1}}
\makeatother
\usepackage{subcaption}
\usepackage{comment}
\usepackage{tikz}
\usetikzlibrary{backgrounds}
\usepackage[bottom=2.5cm,top=2.5cm, left=2.5cm, right=2.5cm]{geometry}

\usepackage{multirow}
\usepackage{array,multirow}

\definecolor{LinkColor}{rgb}{0,0,0} 
\usepackage[colorlinks=true,linkcolor=LinkColor,citecolor=LinkColor,urlcolor= LinkColor, naturalnames, hyperindex, pdfstartview=FitH, bookmarksnumbered, plainpages]{hyperref}
\usepackage{cleveref}

\makeatletter
\newcommand{\slunlhd}{%
	\mathrel{\mathpalette\sl@unlhd\relax}%
}

\newcommand{\sl@unlhd}[2]{%
	\sbox\z@{$#1\lhd$}%
	\sbox\tw@{$#1\leqslant$}%
	\dimen@=\ht\tw@
	\advance\dimen@-\ht\z@
	\ifx#1\displaystyle
	\advance\dimen@ .2pt
	\else
	\ifx#1\textstyle
	\advance\dimen@ .2pt
	\fi
	\fi
	\ooalign{\raisebox{\dimen@}{$\m@th#1\lhd$}\cr$\m@th#1\leqslant$\cr}%
}
\makeatother

\newtheorem{innercustomthm}{Theorem}[section]

\crefname{innercustomthm}{Theorem}{Theorems}

\newtheorem{theorem}{Theorem}[section]
\newtheorem{corollary}[theorem]{Corollary}
\newtheorem{lemma}[theorem]{Lemma}

\newtheorem{proposition}[theorem]{Proposition}

\theoremstyle{definition}
\newtheorem{definition}[theorem]{Definition}
\newtheorem{remark}[theorem]{Remark}
\newtheorem{example}[theorem]{Example}

\newcommand{\Aut}{\operatorname{Aut}}

\newcommand{\Tor}{\operatorname{Tor}}

\newcommand{\N}{\textup{N}}

\author[S. Antil]{Seema Antil}
\address{(Seema Antil) Department of Mathematics, 
	Indian Institute of Technology Ropar, Ropar (Punjab)-140001, India.}
	\email{\href{mailto:seema.22maz0010@ma.iitrpr.ac.in}{seema.22maz0010@iitrpr.ac.in}}
	\author[S. Chahal]{Seema Chahal}
	\address{(Seema Chahal) Department of Mathematics, Indian Institute of Technology Roorkee, Roorkee (Uttarakhand)-247667, India.}
	\email{\href{mailto:seema_r@ma.iitr.ac.in}{seema\_r@ma.iitr.ac.in}}

\author[M. Khan]{Manju Khan}
\address{(Manju Khan) Department of Mathematics, 
	Indian Institute of Technology Ropar, Ropar (Punjab)-140001, India.}
\email{\href{mailto:manju@.iitrpr.ac.in}{manju@iitrpr.ac.in}}
\author[S. Maheshwary]{Sugandha Maheshwary}
\address{(Sugandha Maheshwary) Indian Institute of Technology Roorkee, Roorkee (Uttarakhand)-247667, India.}
\email{\href{mailto:msugandha@ma.iitr.ac.in}{msugandha@ma.iitr.ac.in}}

\thanks{The first-named author gratefully acknowledges the financial support provided by the Ministry of  Education, Government of India. The first author also acknowledges the partial support from the FIST program of the Department of Science and Technology, Government of India, Reference No. SR/FST/MS-I/2018/22(C). The third-named author acknowledges the financial support of ANRF (File no. MTR/2022/000616. The fourth-named author gratefully acknowledges the support by Science  \& Engineering Research Board (SERB),  DST (Department of Science and Technology), India (SRG/2023/000180)}



\date{}

\title{A study of skew-polycyclic codes over a non-chain ring}

\begin{document}
	\maketitle
	
	\begin{abstract}
		For a prime \(p\) and a positive integer \(m\), let \(\mathbb{F}_{p^m}\) be the finite field of cardinality \(p^m\), and let
		$
		R_{u^2,v^2,p^m}
		=\mathbb{F}_{p^m}+u\mathbb{F}_{p^m}+v\mathbb{F}_{p^m}
		+uv\mathbb{F}_{p^m},
		~ u^2=v^2=0,\ uv=vu,
		$
		be a finite non-chain ring.  In this paper, we study skew polycyclic codes of length \(lj\) associated with \(f(x)^j\), where \(f(x)\) is a central polynomial of degree \(l\) in $R_{u^2, v^2, p^m}[x; \Theta],$ where $\Theta$ being an automorphism of  \(R_{u^2,v^2,p^m}\). We describe these codes, characterize free skew polycyclic codes, and determine their ranks. 
	Under suitable centrality assumptions, we decompose the quotient ring associated with \(x^{np^s}-\lambda\), where \(\gcd(n,p)=1\) and \(\Theta(\lambda)=\lambda\). This reduces the study of skew \((\lambda,\Theta)\)-constacyclic codes of length \(np^s\) to the study of left ideals of
		$\frac{R_{u^2,v^2,p^m}[x;\Theta]}{\langle f(x)^j\rangle}, 
		$ where  \(f(x)\) is a central irreducible divisor  of degree \(l\) 	of \(x^{np^s}-\lambda\), for an invertible element  \(\lambda\in R_{u^2,v^2,p^m}\) and \(j\in\mathbb{N}\). 
	We then apply these results to skew \((\lambda,\Theta)\)-constacyclic codes of length \(p^s\) for different classes of units \(\lambda\). Several examples are presented to illustrate the theory and to obtain optimal codes. Finally, when \(\Theta\) is the identity automorphism, we study constacyclic codes of length \(np^s\) over \(R_{u^2,v^2,p^m}\), according as \(x^n-\alpha_0\) is irreducible or reducible over \(\mathbb{F}_{p^m}\). These results extend the work of \cite{CCDF18} and \cite{ZTG18} on constacyclic codes of length \(np^s\) over \(\mathbb{F}_{p^m}+u\mathbb{F}_{p^m}\) to the finite non-chain ring \(R_{u^2,v^2,p^m}\).
	\end{abstract}

	\section{Introduction}
The class of constacyclic codes has attracted considerable attention due to its rich algebraic structure and numerous applications in coding theory.  In particular, the classification of such codes plays a crucial role in understanding their algebraic structure. In the last two decades, considerable attention has been devoted to studying constacyclic codes over various finite rings. These codes can be utilized in cryptography, data transmission, data compression, and storage systems, where they play an important role in the detection and correction of errors in various communication channels. 

Although extensive research has been carried out on constacyclic codes over finite rings, a complete classification is generally difficult and is known only for certain lengths over specific finite fields and finite chain rings.  In this context, Zhao et al. \cite{ZTG18} and Cao et al. \cite{CCDF18} determined the constacyclic codes of length $np^s$ over $\mathbb{F}_{p^m}+u\mathbb{F}_{p^m}.$  Also 
A non-commutative generalization of cyclic and constacyclic codes is
obtained by using skew polynomial rings. Boucher et al.~\cite{BGU07}
introduced skew cyclic codes over finite fields by considering the skew
polynomial ring \(\mathbb{F}_{p^m}[x;\theta]\), where \(\theta\) is an
automorphism of \(\mathbb{F}_{p^m}\). Later, skew constacyclic and skew
polycyclic codes were studied over several classes of finite rings; see
\cite{JLU12,HS23,RPM26,CAMK26,TS26,BMO26}. Since skew polynomial rings are
usually non-commutative, their factorization theory is richer than that of
ordinary polynomial rings. This provides a useful framework for obtaining
new families of codes.

Beyond finite chain rings, researchers have also studied codes over
finite non-chain rings. Yildiz and Karadeniz~\cite{YK11} studied cyclic codes over the non-chain ring
$\mathbb{F}_2+u\mathbb{F}_2+v\mathbb{F}_2+uv\mathbb{F}_2,$
where \(u^2=v^2=0, uv=vu\). 
Dougherty et al., \cite{DKY12} studied cyclic codes over the family of rings $R_k=\mathbb{F}_2[u_1,\ldots,u_k]/\langle u_i^2, u_i u_j-u_j u_i \mid 1\leq i,j\leq k\rangle,$
 which includes non-chain rings for \(k\geq 2\).
Kewat et al.~\cite{KGP15} determined the
algebraic structure of cyclic codes over the ring
$\mathbb{Z}_p+u\mathbb{Z}_p+v\mathbb{Z}_p+uv\mathbb{Z}_p$. Later,
Dinh et al.~\cite{DKKY20} investigated the algebraic structure of
constacyclic codes of length $p^s$ over the finite non-chain ring
$\mathbb{F}_{p^m}+u\mathbb{F}_{p^m}+v\mathbb{F}_{p^m}+uv\mathbb{F}_{p^m}$.

 Motivated by these works, in this paper we classify and investigate the algebraic structure of skew polycyclic codes over a  finite non-chain ring $\mathbb{F}_{p^m}+u\mathbb{F}_{p^m}+v\mathbb{F}_{p^m}+uv\mathbb{F}_{p^m}$ with  $u^2=0,$ $v^2=0,$ $uv=vu,$ for any prime $p$ and positive integer $m.$ For notational convenience, denote $\mathbb{F}_{p^m}+u\mathbb{F}_{p^m}+v\mathbb{F}_{p^m}+uv\mathbb{F}_{p^m}$ by $R_{u^2, v^2, p^m}$ and $R_{u^2, v^2, p^m}[x, \Theta]$ is the skew polynomial ring, with $\Theta$ is an automorphism of $R_{u^2, v^2, p^m}$ satisfying $xa=\Theta(a)x$ for all $a \in R_{u^2, v^2, p^m}.$ 
 It is assumed that \(f(x)\) is a central polynomial in \(R_{u^2,v^2,p^m}[x;\Theta]\). For \(j\in\mathbb{N}\), we investigate skew \((f^j,\Theta)\)-polycyclic codes, which correspond to left ideals of the quotient ring
 $
 \frac{R_{u^2,v^2,p^m}[x;\Theta]}{\langle f(x)^j\rangle}.
 $
 We further study the structure and algebraic properties of free skew polycyclic codes over \(R_{u^2,v^2,p^m}\), thereby extending the results of \cite{KGP15} to the skew polycyclic setting. As an application, when \(f(x)\) is a central irreducible divisor of \(x^{np^s}-\lambda\) of degree \(l\), where $\lambda$ is a unit in \( R_{u^2,v^2,p^m}\) and \(\Theta(\lambda)=\lambda\), the study of skew constacyclic codes of length \(np^s\) over \(R_{u^2,v^2,p^m}\) can be reduced to the study of skew \((f^j,\Theta)\)-polycyclic codes of length \(jl\).
  In particular, when $n=1$, we investigate all skew $(\lambda_i,\Theta)$-constacyclic codes, where $\lambda_i$ is a unit element of $R_{u^2, v^2, p^m}$, by classifying them into five categories and studying their algebraic structures. Furthermore, we provide several examples to illustrate the obtained results. Lastly, when $\Theta$ is an identity automorphism, and $x^n-\alpha_0$ is irreducible over $\mathbb{F}_{p^m}$, where $\alpha_0 \in \mathbb{F}_{p^m}$, then we study $\alpha$-constacyclic codes of length $np^s$ over $R_{u^2, v^2, p^m}$. When $x^n-\alpha_0$ is reducible over $\mathbb{F}_{p^m}$, then we discuss the decomposition of $x^{np^s}-\lambda$ over $R_{u^2, v^2, p^m}$. 
  
  The paper is organized as follows:  In Section $2$,  we set basic terminology. In Section $3$, we study the skew $(f^j, \Theta)$-polycyclic codes of length $jl$ over $R_{u^2, v^2, p^m}$,  where $p$ is a prime and $j,l, m \in \mathbb{N}.$ 
 We further investigate the structure of free skew polycyclic codes over $R_{u^2,v^2,p^m}$ and their algebraic properties and hence generalize the results in \cite{KGP15}.  Section~4 investigates
 skew \((\lambda,\Theta)\)-constacyclic codes of length \(p^s\) for the
 relevant unit cases and provides examples to illustrate the theoretical results. In the last section,  we study constacyclic codes of length \(np^s\) over \(R_{u^2,v^2,p^m}\), according as \(x^n-\alpha_0\) is irreducible or reducible over \(\mathbb{F}_{p^m}\). These results extend the results of \cite{CCDF18} and \cite{ZTG18} on constacyclic codes of length \(np^s\) over \(\mathbb{F}_{p^m}+u\mathbb{F}_{p^m}\) to the finite non-chain ring \(R_{u^2,v^2,p^m}\).

	\section{Preliminaries}
As stated above,  \(\mathbb{F}_{p^m}\) denotes the field of cardinality \(p^m\), where $p$ is a prime and  \(m\in \mathbb{N}\). The ring
$\mathbb{F}_{p^m}+u\mathbb{F}_{p^m}+v\mathbb{F}_{p^m}
+uv\mathbb{F}_{p^m}$, where $
~ u^2=v^2=0,\ uv=vu,
$ is denoted by $R_{u^2,v^2,p^m}$. 

	The sets of units of
	\(\mathbb{F}_{p^m}\) and \(R_{u^2,v^2,p^m}\) are denoted by
	\(\mathbb{F}_{p^m}^{*}\) and \(R_{u^2,v^2,p^m}^{*}\), respectively.
	The set of automorphisms of 	\(\mathbb{F}_{p^m}\) and \(R_{u^2,v^2,p^m}\) are denoted by $\Aut(\mathbb{F}_{p^m})$ and $\Aut(R_{u^2,v^2,p^m})$,  respectively.
	For  $\Theta \in \Aut(R_{u^2,v^2,p^m})$,  
		$R_{u^2,v^2,p^m}[x;\Theta]$ consists of  polynomials
		$a_0+a_1x+\cdots+a_nx^n,
		~~\text{where}~ a_i\in R_{u^2,v^2,p^m},
		$ and forms a ring 
		with the usual addition and multiplication determined by the rule
		$xa=\Theta(a)x,~ \forall\, a\in R_{u^2,v^2,p^m}
		$. 	The ring  $R_{u^2,v^2,p^m}[x;\Theta]$ is called as skew polynomial ring and 
		the elements of $R_{u^2,v^2,p^m}[x;\Theta]$ are
	 called as skew polynomials.
We say that $f(x)$ is a right divisor (respectively, left divisor) of
	$g(x)$ in $R_{u^2,v^2,p^m}[x;\Theta]$, and write
	$f(x)\mid_r g(x)$ (respectively, $f(x)\mid_l g(x)$), if there exists
	$h(x)\in R_{u^2,v^2,p^m}[x;\Theta]$ such that $g(x)=h(x)f(x)$ (respectively,  $g(x)=f(x)h(x)).$
	Let $f(x),g(x)\in R_{u^2,v^2,p^m}[x;\Theta]$. The greatest common right divisor of $f(x)$ and $g(x)$ is the monic polynomial
		$d_r(x)\in R_{u^2,v^2,p^m}[x;\Theta]$ such that $d_r(x)\mid_r f(x),$ $d_r(x)\mid_r g(x),$
		and for any $d_r'(x)\in R_{u^2,v^2,p^m}[x;\Theta]$ satisfying $d_r'(x)\mid_r f(x)
		~\text{and}~
		d_r'(x)\mid_r g(x),$ we have $d_r'(x)\mid_r d_r(x)$. We denote $d_r(x)$ by
		$\gcd_r(f(x),~g(x))$.

By \cite{McD74}, we shall use the right division algorithm in the skew
polynomial ring \(R_{u^2,v^2,p^m}[x;\Theta]\). More precisely, if
\(f(x), ~g(x)\in R_{u^2,v^2,p^m}[x;\Theta]\) and the leading coefficient of
\(f(x)\) is invertible, then there exist
\(q(x),r(x)\in R_{u^2,v^2,p^m}[x;\Theta]\) such that
\[
g(x)=q(x)f(x)+r(x),
\]
where either \(r(x)=0\) or \(\deg(r(x))<\deg(f(x))\).

	The following proposition is adapted from \cite{JLU12,BMO26} and its proof remains valid for the finite non-chain ring \(R_{u^2,v^2,p^m}\).
	\begin{proposition}\label{central polynomials} Let $f(x)=x^n-\sum_{i=0}^{n-1}a_i x^i \in R_{u^2,v^2,p^m}[x;\Theta], $ then \(f(x)\) is central if and only if \begin{enumerate} \item \(\Theta(a_i)=a_i\) for all \(0\le i\le n-1\); \item \(a_i r=\Theta^i(r)a_i\) for all \(r\in R_{u^2,v^2,p^m}\) and \(0\le i\le n-1\); \item \(\Theta^n=\mathrm{Id}_{R_{u^2,v^2,p^m}}\). \end{enumerate} In particular, for \(\lambda\in R_{u^2,v^2,p^m}^{*}\), the polynomial \(x^n-\lambda\) is central in \(R_{u^2,v^2,p^m}[x;\Theta]\) if and only if \(n\) is a multiple of the order of \(\Theta\) and \(\Theta(\lambda)=\lambda\). Equivalently, \(\langle x^n-\lambda\rangle\) is a two-sided ideal of \(R_{u^2,v^2,p^m}[x;\Theta]\). \end{proposition}

	A code $\mathcal{C}$ of length $n$ over $R_{u^2,v^2,p^m}$ is a non-empty subset of
	$R_{u^2,v^2,p^m}^{\,n}$. The ring $R_{u^2,v^2,p^m}$ is called the alphabet of
	$\mathcal{C}$ and the elements of $\mathcal{C}$ are called codewords.  The notion of skew polycyclic codes over finite chain rings was introduced in \cite{BMO26}. Since the definition depends only on the skew multiplication rule \(xa=\Theta(a)x\), it is also valid over the finite non-chain ring \(R_{u^2,v^2,p^m}\). Let
$f(x)=x^n-\sum_{i=0}^{n-1}a_i x^i \in R_{u^2,v^2,p^m}[x;\Theta]$
be a monic polynomial of degree $n$. For a given automorphism $\Theta$ of $R_{u^2,v^2,p^m}$, a linear code $\mathcal{C}$ over $R_{u^2,v^2,p^m}$ is said to be skew $(f,\Theta)$-polycyclic code if, for every $(c_0,c_1,\ldots,c_{n-1})\in\mathcal{C},$ the element 
 $\bigl(
\Theta(c_{n-1})a_0,\,
\Theta(c_0)+\Theta(c_{n-1})a_1,\,
\ldots,\,
\Theta(c_{n-2})+\Theta(c_{n-1})a_{n-1}
\bigr)$ $\in \mathcal{C}$.  If $f(x)$ is central, then $\mathcal{C}$ is a skew $(f,\Theta)$-polycyclic code if and only if $\mathcal{C}$ is a left ideal of
$\frac{R_{u^2,v^2,p^m}[x;\Theta]}{\langle f(x)\rangle}.$
More precisely, a linear code $\mathcal{C}$ over $R_{u^2,v^2,p^m}$ is said to be skew $(\lambda,\Theta)$-constacyclic if it is invariant under the $\Theta$-constacyclic shift $\rho_{\Theta}:R_{u^2,v^2,p^m}^{\,n}\longrightarrow R_{u^2,v^2,p^m}^{\,n}$
defined by $\rho_{\Theta}(a_0,a_1,\ldots,a_{n-1})
=
\bigl(
\Theta(\lambda a_{n-1}),
\Theta(a_0),
\Theta(a_1),
\ldots,
\Theta(a_{n-2})
\bigr).$ Equivalently, $\mathcal{C}$ is a skew $(\lambda,\Theta)$-constacyclic code if and only if it is a left ideal of ~$\frac{R_{u^2,v^2,p^m}[x;\Theta]}{\langle x^n-\lambda\rangle}.$ In particular, skew $(\lambda,\Theta)$-constacyclic codes of length $np^s$ are precisely skew $(f,\Theta)$-polycyclic codes corresponding to the polynomial $f(x)=x^{np^s}-\lambda,$ $\lambda\in R_{u^2,v^2,p^m}^{*}.$
 Furthermore, when $\Theta$ is an identity automorphism of $R_{u^2,v^2,p^m}$, the skew $(\lambda,\Theta)$-constacyclic codes reduce to the classical $\lambda$-constacyclic codes.

We now introduce the Gray map, which allows us to study codes over
$R_{u^2,v^2,p^m}$ through their images over the finite field $\mathbb{F}_{p^m}$.
	\begin{definition}
		The Gray map
		\[
		\phi:R_{u^2,v^2,p^m}\longrightarrow \mathbb{F}_{p^m}^{\,4}
		\]
		is defined by
		\[
		\phi(a+ub+vc+uvd)
		=
		(d,\;c+d,\;b+d,\;a+b+c+d),
		\]
		for all \(a,b,c,d\in\mathbb{F}_{p^m}\).
	\end{definition}
		The Gray map naturally extends to $R_{u^2,v^2,p^m}^n$ as distance-preserving isometry
		\[
	\Phi:R_{u^2,v^2,p^m}^n \longrightarrow \mathbb{F}_{p^m}^{\,4n}
	\]
	is defined by
	\[
	\Phi(a_1, a_2, \cdots, a_n)
	=
	(\phi(a_1), \phi(a_2), \cdots, \phi(a_n)),
	\]
	for all \(a_i \in R_{u^2, v^2, p^m}\).
	
	By linearity of the map $\Phi$, we have the following lemma.
	
	\begin{lemma} {\cite[Theorem 2.3]{KGP15}}
		If $\mathcal{C}$ is a linear code over the ring $R_{u^2,v^2,p^m}$ of length $n$, size $p^k$ and minimum distance $d$, then $\Phi(\mathcal{C})$ is a linear code with parameters $[4n,k,d].$
	\end{lemma}

	\section{Skew Polycyclic Codes over $R_{u^2,v^2,p^m}$}
	In this section, we study skew $(f^j, \Theta)$-polycyclic codes over $R_{u^2,v^2,p^m}$ , which correspond to left ideals of the quotient ring $R_{u^2, v^2, p^m}[x, \Theta]/ \langle f(x)^j \rangle,$ for $j \in \mathbb{N}, $ where \(f(x)\) is a central polynomial of degree \(l\) in $R_{u^2, v^2, p^m}[x; \Theta]$. Consequently, 
the	skew $(\lambda,\Theta)$-constacyclic codes of length $np^s$, where $(n,p)=1$ and $s\geq0$ are studied.\\

We begin by describing the automorphisms of the ring $R_{u^2,v^2,p^m}$.
\begin{lemma}\label{auto_R_2_2}
	Let $\theta\in\Aut(\mathbb{F}_{p^m})$.
	For $\alpha_1,  \alpha_2, \beta_1, \beta_2, \gamma_1, \gamma_2
	\in\mathbb{F}_{p^m}$, let

	\[
	\Theta:
	R_{u^2,v^2,p^m}\longrightarrow R_{u^2,v^2,p^m}
	\]
	be a ring homomorphism such that
	$
	\Theta(a)=\theta(a)
	~\text{for all }a\in\mathbb{F}_{p^m},
	$
	\[
	\Theta(u)=\alpha_1u+\beta_1v+\gamma_1uv
~~	\text{and}~~
	\Theta(v)=\alpha_2u+\beta_2v+\gamma_2uv.
	\]
	Then $	\Theta \in\Aut(R_{u^2,v^2,p^m})$ if and only if
	$
	\alpha_1 \beta_2-\alpha_2 \beta_1 \neq0$ and $\alpha_1 \beta_2+\alpha_2 \beta_1  \neq0.
	$ Furthermore, any automorphism of $R_{u^2,v^2,p^m}$ is of this form.

\end{lemma}

\begin{proof}
We first check that if 	$
	\alpha_1 \beta_2-\alpha_2 \beta_1 \neq0
	\quad\text{and}\quad
	\alpha_1 \beta_2+\alpha_2 \beta_1  \neq0$ then, 
	$\Theta  \in \Aut( R_{u^2,v^2,p^m})$.
	To check injectivity of $\Theta$, 	suppose
	$\Theta(a_0+a_1u+a_2v+a_3uv)=0, ~ a_i \in \mathbb{F}_{p^m},~ 0 \leq i \leq 3.$
	By	comparing the coefficients of $1,u,v,uv$ on both sides, we get
	\[
	\theta(a_0)=0,
	\]
	\[
	\theta(a_1)\alpha_1+\theta(a_2)\alpha_2=0,
	\qquad
	\theta(a_1)\beta_1+\theta(a_2)\beta_2=0,
	\]
	and
	\[
	\theta(a_1)\gamma_1+\theta(a_2)\gamma_2
	+\theta(a_3)(\alpha_1 \beta_2+\alpha_2 \beta_1)=0.
	\]
	Since
	$\alpha_1 \beta_2-\alpha_2 \beta_1
\neq0,
	$
	the middle two equations imply $a_1=a_2=0$. Since
	$
\alpha_1 \beta_2+\alpha_2 \beta_1\neq0,
	$
	the last equation gives $a_3=0$. Also $a_0=0$. Hence $\ker\Theta=\{0\}$.
	Since $R_{u^2,v^2,p^m}$ is finite, $\Theta$ is bijective, and therefore
	$\Theta\in\Aut(R_{u^2,v^2,p^m})$.\\
	Further, suppose $\Phi \in \Aut( R_{u^2,v^2,p^m})$. Since $\mathbb{F}_{p^m}$ is the residue field of $ R_{u^2,v^2,p^m}$, the restriction of $\Phi$ to $\mathbb{F}_{p^m}$ gives a field automorphism $\theta \in \Aut(\mathbb{F}_{p^m})$.
	As $\langle u, v \rangle$ is the unique maximal ideal of $ R_{u^2,v^2,p^m}$, it must map to itself under any ring automorphism of $ R_{u^2,v^2,p^m}$.
	Hence \[ \Phi(u)=\alpha_1^{'} u+\beta_1^{'} v+\gamma_1^{'} uv ~~\text{and}~~ \Phi(v)=\alpha_2^{'} u+\beta_2^{'}v+\gamma_2^{'} uv \] for some $\alpha_1^{'}, \alpha_2^{'}, \beta_1^{'}, \beta_2^{'}, \gamma_1^{'}, \gamma_2^{'} \in\mathbb{F}_{p^m}$.
	Then $
	\Phi(uv) = \Phi(u)\Phi(v)
	= (\alpha_1^{'} \beta_2^{'}+\alpha_2^{'} \beta_1^{'}) uv,
	$
	but $\Phi$ is a one-one homomorphism, hence $ \alpha_1^{'} \beta_2^{'}+\alpha_2^{'} \beta_1^{'} \neq 0$. 
	Suppose $a_0 + a_1 u + a_2 v + a_3 uv \in \ker \Phi ,$ for some $a_i \in \mathbb{F}_{p^m},$ we have, 
	\begin{equation}\label{eq lemma 1}
		\theta(a_1)\alpha_1^{'} + \theta(a_2)\alpha_2^{'} = 0,
	\end{equation}
	\begin{equation}\label{eq lemma 2}
		\theta(a_1)\beta_1^{'} + \theta(a_2)\beta_2^{'}= 0,
	\end{equation}
	\begin{equation}\label{eq lemma 3}
		\theta(a_1)(\gamma_1^{'}) + \theta(a_2)(\gamma_2^{'}) 
		+ \theta(a_3)(\alpha_1^{'} \beta_2^{'}+\alpha_2^{'} \beta_1^{'}) = 0.
	\end{equation}
	By solving,  (\ref{eq lemma 1}) and (\ref{eq lemma 2}) we get,
	$\theta(a_1)[\alpha_1^{'} \beta_2^{'}-\alpha_2^{'} \beta_1^{'}]=0$ and $\theta(a_2)[\alpha_1^{'} \beta_2^{'}-\alpha_2^{'} \beta_1^{'}]=0.$
	Now, if $\alpha_1^{'} \beta_2^{'}-\alpha_2^{'} \beta_1^{'}=0$ then $\Phi$ is no longer one-one. Hence, (\ref{eq lemma 1}),  (\ref{eq lemma 2}) and  (\ref{eq lemma 3}) have a trivial solution if and only if $\alpha_1^{'} \beta_2^{'}-\alpha_2^{'} \beta_1^{'}\neq 0$ and $\alpha_1^{'} \beta_2^{'}+\alpha_2^{'} \beta_1^{'} \neq 0.$
	Therefore, $\Phi$ is of the desired form.
\end{proof}

For an automorphism $\Theta\in\Aut(R_{u^2,v^2,p^m})$, let
$\theta=\Theta|_{\mathbb{F}_{p^m}}$ be its restriction to the base field.
The canonical projection
\begin{align*}
\mu:R_{u^2,v^2,p^m}&\longrightarrow \mathbb{F}_{p^m},\\
a_0+a_1u+a_2v+a_3uv&\longmapsto a_0,
\end{align*}

extends coefficient-wise to
\[
\mu:\mathcal{R}^{jl}_{u^2,v^2,f}=\frac{R_{u^2,v^2,p^m}[x;\Theta]}{\langle f(x)^j\rangle}
\longrightarrow
\mathcal{R}^{jl}_{\mu(f)}
=
\frac{\mathbb{F}_{p^m}[x;\theta]}{\langle \mu(f(x))^j\rangle},
\] where $l$ is the degree of $f(x)$.
Similarly, the projections
\begin{align*}
\pi_v:R_{u^2,v^2,p^m}\longrightarrow &R_{u^2,p^m},\\
~
a_0+a_1u+a_2v+a_3uv\longmapsto& a_0+a_1u,
\end{align*}

and
\begin{align*}
\pi_u:R_{u^2,v^2,p^m}\longrightarrow& R_{v^2,p^m},
\\
a_0+a_1u+a_2v+a_3uv\longmapsto &a_0+a_2v,
\end{align*}

extend coefficient-wise to
\[
\pi_v:\mathcal{R}^{jl}_{u^2,v^2,f}
\longrightarrow
\mathcal{R}^{jl}_{u^2,\pi_v(f)}
=
\frac{R_{u^2,p^m}[x;\Theta_v]}{\langle \pi_v(f(x))^j\rangle},
\]
and
\[
\pi_u:\mathcal{R}^{jl}_{u^2,v^2,f}
\longrightarrow
\mathcal{R}^{jl}_{v^2,\pi_u(f)}
=
\frac{R_{v^2,p^m}[x;\Theta_u]}{\langle \pi_u(f(x))^j\rangle},
\]
where $\Theta_v$ and $\Theta_u$ denote the automorphisms induced by $\Theta$
on $R_{u^2,p^m}$ and $R_{v^2,p^m}$, respectively.

Now let $\mathcal{I}$ be a left ideal of
$
\mathcal{R}^{jl}_{u^2,v^2,f}.$
We use the following notation.

\begin{enumerate}
	\item The $v$-torsion of $\mathcal{I}$ is
	\[
	\Tor_v(\mathcal{I})
	=
	\pi_v\left(
	\left\{
	h(x)\in \mathcal{R}^{jl}_{u^2,v^2,f}
	\;\middle|\;
	v h(x)\in \mathcal{I}
	\right\}
	\right)
	\subseteq
	\mathcal{R}^{jl}_{u^2,\pi_v(f)}
	\] and 	the $u$-torsion of $		\Tor_v(\mathcal{I}) $ is 
	\[
		\Tor_u	(	\Tor_v(\mathcal{I}) )
	= \pi_u\left( \left\{\, h(x) \in  \mathcal{R}^{jl}_{ u^2, ~\pi_v(f)} \;\middle|\; u\,h(x) \in 	\Tor_v(\mathcal{I}) \right\}\right).
	\]
	
	\item The $u$-torsion of $\mathcal{I}$ is
	\[
		\Tor_u(\mathcal{I})
	=
	\pi_u\left(
	\left\{
	h(x)\in \mathcal{R}^{jl}_{u^2,v^2,f}
	\;\middle|\;
	u h(x)\in \mathcal{I}
	\right\}
	\right)
	\subseteq
	\mathcal{R}^{jl}_{v^2,\pi_u(f)}.
	\]
	
	\item The $uv$-torsion of $\mathcal{I}$ is
	\[
		\Tor_{uv}(\mathcal{I})
	=
	\mu\left(
	\left\{
	h(x)\in \mathcal{R}^{jl}_{u^2,v^2,f}
	\;\middle|\;
	uv h(x)\in \mathcal{I}
	\right\}
	\right)
	\subseteq
	\mathcal{R}^{jl}_{\mu(f)}.
	\]
\end{enumerate}

It is easy to observe that $	\mathrm{Tor}_{uv}(\mathcal{I}) =	\mathrm{Tor}_u	(\mathrm{Tor}_v(\mathcal{I}) ).$\\

We are now in a position to study skew $(f^j,\Theta)$-polycyclic codes of length $jl$ over $R_{u^2,v^2,p^m}$, where $p$ is a prime, $j,l,m\in\mathbb{N}$, and $\deg f(x)=l$. These codes correspond to the left ideals of \[ \mathcal{R}^{jl}_{u^2,v^2,f} = \frac{R_{u^2,v^2,p^m}[x;\Theta]}{\langle f(x)^j\rangle}. \] We assume that $f(x)$ is a central polynomial in $R_{u^2,v^2,p^m}[x;\Theta]$ and denote	the set of all monic proper  divisors of the central polynomial $\mu(f(x)^j)$  by
$\mathcal{B}^{jl}_{ \mu{(f)}}$, where $\mu(f(x))$ is a central polynomial in   $\mathbb{F}_{p^m}[x;\theta].$
	\subsection{Skew-Polycyclic Codes Over $R_{u^2, v^2, p^m}$}
		\begin{theorem}\label{ideals R_{u^2, v^2, p^m}}
		Every left ideal of 
		$\mathcal{R}^{jl}_{u^2,v^2,f}$
		is of the form
		\begin{align*}
			\mathcal{I} 
			= &\mathcal{R}^{jl}_{u^2,v^2,f}
			\big(f_1(x)+u f_{1,2}(x)+v f_{1,3}(x)+uv f_{1,4}(x)\big) 
			+\\& \mathcal{R}^{jl}_{u^2,v^2,f}
			\big(u f_2(x)+v f_{2,3}(x)+uv f_{2,4}(x)\big) 
			+\\& \mathcal{R}^{jl}_{u^2,v^2,f}
			\big(v f_3(x)+uv f_{3,4}(x)\big) 
		+\\& \mathcal{R}^{jl}_{u^2,v^2,f}
			\big(uv f_4(x)\big),
		\end{align*}	where $f_{i,j}(x)\in
		\mathcal{R}^{jl}_{\mu(f)},$ for every $i,j$ and  $f_i(x)$ is either $0$ or $f_i(x)\in   \mathcal{B}^{jl}_{\mu(f)}$, for every $i$. In the later case, $
		f_4(x)\mid_r f_i(x),
		f_i(x)\mid_r f_1(x)$ and
		$\deg(f_{i,j}(x))<\deg(f_j(x)).
		$
		Further, the polynomials $f_{i,j}(x)$ satisfying these conditions are unique.
	\end{theorem}
		\begin{proof}
		Any element $c(x)\in \mathcal I$ can be written as
		$
		c(x)=g_0(x)+u g_1(x)+v g_2(x)+uv g_3(x),
		$
		where $g_i(x)\in \mathcal{R}^{jl}_{\mu(f)}$.
			Consider the restriction of the coefficient-wise projection
		$
		\pi_v|_{\mathcal I}:
		\mathcal I
		\longrightarrow
		\mathcal{R}^{jl}_{u^2,\pi_v(f)}.
		$
		Since $\pi_v(\mathcal I)$ is a left ideal of
		$\mathcal{R}^{jl}_{u^2,\pi_v(f)}$, by \cite[Theorem 3.5]{HS23}, we have
		\[
		\pi_v(\mathcal I)
		=
		\mathcal{R}^{jl}_{u^2,\pi_v(f)}
		\big(f_1(x)+u f_{1,2}(x)\big)
		+
		\mathcal{R}^{jl}_{u^2,\pi_v(f)}
		\big(u f_2(x)\big).
		\]
		Hence there exist elements
		$
		G_1=
		f_1(x)+u f_{1,2}(x)+v f_{1,3}(x)+uv f_{1,4}(x)
		$
		and
		$
		G_2=
		u f_2(x)+v f_{2,3}(x)+uv f_{2,4}(x)
		$
		in $\mathcal I$ such that
		$
		\pi_v(G_1)=f_1(x)+u f_{1,2}(x),
		~
		\pi_v(G_2)=u f_2(x).
		$
		
		Now define
		$
		J=
		\left\{
		\ell(x)\in \mathcal{R}^{jl}_{u^2,\pi_v(f)}
		\;\middle|\;
		v\ell(x)\in \mathcal I
		\right\}.
		$
		Then $J=	\Tor_v(\mathcal I)$ and $J$ is a left ideal of
		$\mathcal{R}^{jl}_{u^2,\pi_v(f)}$. Again, by \cite[Theorem 3.5]{HS23},
		$
		J
		=
		\mathcal{R}^{jl}_{u^2,\pi_v(f)}
		\big(f_3(x)+u f_{3,4}(x)\big)
		+
		\mathcal{R}^{jl}_{u^2,\pi_v(f)}
		\big(u f_4(x)\big).
		$
		Therefore,
		$
		\ker(\pi_v|_{\mathcal I})
		=
		vJ
		\subseteq \mathcal I.
		$
		Thus we obtain two additional generators
		$
		G_3=v f_3(x)+uv f_{3,4}(x)
	~	\text{and}~
	G_4=uv f_4(x).
		$
		
		The exact sequence
		$
		0
		\longrightarrow
		\ker(\pi_v|_{\mathcal I})
		\longrightarrow
		\mathcal I
		\overset{\pi_v}{\longrightarrow}
		\pi_v(\mathcal I)
		\longrightarrow
		0
		$
		implies that
	
	\begin{equation}\label{eq ideal}
		\mathcal I
		=
		\mathcal{R}^{jl}_{u^2,v^2,f}(G_1)
		+
		\mathcal{R}^{jl}_{u^2,v^2,f}(G_2)
		+
		\mathcal{R}^{jl}_{u^2,v^2,f}(G_3)
		+
		\mathcal{R}^{jl}_{u^2,v^2,f}(G_4).
	\end{equation}
	Further,
	$
	f_1(x),f_2(x)
	\in
	\Tor_{uv}(\mathcal I)
	=
	\mathcal{R}^{jl}_{\mu(f)}(f_4(x)),
	$
	which yields
	$
	f_4(x)\mid_r f_1(x),
	~
	f_4(x)\mid_r f_2(x).
	$
	Since
	$
	f_2(x)\mid_r f_1(x)
	~ \text{and} ~
	f_4(x)\mid_r f_3(x),
	$
	the stated divisibility conditions follow.
	
	If $f_j(x)\neq 0$, then by right division,
	$
	g_{i,j}(x)
	=
	q_{i,j}(x)f_j(x)+f_{i,j}(x),
	$
	where
	$
	\deg(f_{i,j}(x))
	<
	\deg(f_j(x)).
	$
	For uniqueness, suppose that
	$f_{i,3}'(x)$ and $f_{i,4}'(x)$ also satisfy \Cref{eq ideal}, then
	$
	v(f_{i,3}(x)-f_{i,3}'(x))
	\in
	\mathcal I,
	$
	which implies
	$
	f_{i,3}(x)-f_{i,3}'(x)
	\in
	\Tor_v(\mathcal I).
	$
	If
	$
	f_{i,3}(x)\neq f_{i,3}'(x),
	$
	then
	$
	\deg(f_3(x))
	\le
	\deg(f_{i,3}(x)-f_{i,3}'(x)),
	$
	contradicting
	$
	\deg(f_{i,3}(x))
	<
	\deg(f_3(x)).
	$
	Hence
	$
	f_{i,3}(x)=f_{i,3}'(x).
	$	Similarly,
	$
	uv(f_{i,4}(x)-f_{i,4}'(x))
	\in
	\mathcal I,
	$
	so that
	$
	f_{i,4}(x)-f_{i,4}'(x)
	\in
	\Tor_{uv}(\mathcal I).
	$
	If
	$
	f_{i,4}(x)\neq f_{i,4}'(x),
	$
	then
	$
	\deg(f_4(x))
	\le
	\deg(f_{i,4}(x)-f_{i,4}'(x)),
	$
	again contradicting
	$
	\deg(f_{i,4}(x))
	<
	\deg(f_4(x)).
	$
	Thus
	$
	f_{i,4}(x)=f_{i,4}'(x).
	$
	\end{proof}

\begin{corollary}
In the ongoing notation,  the following relations hold in
		$\mathcal{R}^{jl}_{\mu(f)}.$
		\begin{enumerate}
			\item 	$f_2(x)\mid_r
			\frac{\mu(f(x))^j}{f_1(x)}f_{1,2}(x),
			$
			\item 	$
			f_3(x)\mid_r
			\frac{\mu(f(x))^j}{f_1(x)}
			\bigg(
			f_{1,3}(x)
			-
			\frac{f_{1,2}(x)}{f_2(x)}f_{2,3}(x)
			\bigg),
			$
			\item $
			f_3(x)\mid_r
			\frac{f_1(x)}{f_2(x)}f_{2,3}(x),
			$
			\item 	$
			f_4(x)\mid_r f_{2,3}(x),
			$
			\item 	$
			f_4(x)\mid_r
			\frac{\mu(f(x))^j}{f_3(x)}f_{3,4}(x),
			$
			\item 	$
			f_4(x)\mid_r
			\frac{\mu(f(x))^j}{f_2(x)}
			\bigg(
			f_{2,4}(x)
			-
			\frac{f_{2,3}(x)}{f_3(x)}f_{3,4}(x)
			\bigg),
			$
			\item 	$
			f_4(x)\mid_r
			\bigg(
			f_{1,2}(x)
			-
			\frac{f_1(x)}{f_3(x)}f_{3,4}(x)
			\bigg),
			$
			\item 	$
			f_4(x)\mid_r
			\bigg(
			f_{1,3}(x)
			-
			\frac{f_1(x)}{f_2(x)}f_{2,4}(x)
			+
			\frac{f_1(x)}{f_2(x)f_3(x)}
			f_{2,3}(x)f_{3,4}(x)
			\bigg).
			$
			
		\end{enumerate}
	\end{corollary}

	\begin{proof}
	\noindent	\begin{enumerate}
			
			\item Since $\frac{\mu(f(x))^j}{f_1(x)}
				\big(
				f_1(x)+uf_{1,2}(x)+vf_{1,3}(x)+uvf_{1,4}(x)
				\big)
				=
				u\frac{\mu(f(x))^j}{f_1(x)}f_{1,2}(x)\linebreak
				+ 
				v\frac{\mu(f(x))^j}{f_1(x)}f_{1,3}(x)
				+
				uv\frac{\mu(f(x))^j}{f_1(x)}f_{1,4}(x)
				\in
				\mathcal I.$
			Hence
			$
			u\frac{\mu(f(x))^j}{f_1(x)}f_{1,2}(x)
			\in
			\pi_v(\mathcal I),
			$
			which gives
			$
			\frac{\mu(f(x))^j}{f_1(x)}f_{1,2}(x)
			\in
			\Tor_u(\pi_v(\mathcal I))
			=
			\mathcal{R}^{jl}_{\mu(f)}(f_2(x)).
			$
			Therefore,
			$
			f_2(x)\mid_r
			\frac{\mu(f(x))^j}{f_1(x)}f_{1,2}(x).$	\item By (1), $\frac{\mu(f(x))^j}{f_1(x)} \big (f_1(x)+u f_{1,2}(x)+v f_{1,3}(x)+uv f_{1,4}(x) \big )- \frac{\mu(f(x))^j}{f_1(x)} \frac{f_{1,2}(x)}{f_2(x)} \big( u f_2(x)+v f_{2,3}(x)+uv f_{2,4}(x)\big) 	\in
			\mathcal I$. This gives $v \frac{\mu(f(x))^j}{f_1(x)} \big( f_{1,3}(x)- \frac{f_{1,2}(x)}{f_2(x)}f_{2,3}(x) \big)
			\in \pi_u(\mathcal{I}).$ This further  gives, $\frac{\mu(f(x))^j}{f_1(x)} \big( f_{1,3}(x)- \frac{f_{1,2}(x)}{f_2(x)}f_{2,3}(x)  \big) \in Tor_v(\pi_u(\mathcal{I}))=
		\mathcal{R}^{jl}_{\mu(f)} \big( f_3(x)).$
			Hence, 	(2) holds.

			\item
			Write $u\big(
				f_1(x)+uf_{1,2}(x)+vf_{1,3}(x)+uvf_{1,4}(x)
				\big)-
					\frac{f_1(x)}{f_2(x)}
				\big(
				uf_2(x)+vf_{2,3}(x)+uvf_{2,4}(x)
				\big)\linebreak
			\text{as}~
				uvf_{1,3}(x)
				-	v\frac{f_1(x)}{f_2(x)}f_{2,3}(x)
					-uv\frac{f_1(x)}{f_2(x)}f_{2,4}(x).$
			Thus, 
			$
			\frac{f_1(x)}{f_2(x)}f_{2,3}(x) 
			\in
			\Tor_v(\pi_u(\mathcal I))
			=
			\mathcal{R}^{jl}_{\mu(f)}(f_3(x)),
			$
			and hence,
			$
			f_3(x)\mid_r
			\frac{f_1(x)}{f_2(x)}f_{2,3}(x).
			$
			
			\item
			Since
			$
			u\big(
			uf_2(x)+vf_{2,3}(x)+uvf_{2,4}(x)
			\big)
			=
			uvf_{2,3}(x),
			$
			we obtain that \linebreak
			$
			f_{2,3}(x)
			\in
			\Tor_{uv}(\mathcal{I})
			=
			\mathcal{R}^{jl}_{\mu(f)}(f_4(x)).
			$	Therefore,
			$
			f_4(x)\mid_r f_{2,3}(x).
			$
				\item
			We have that 
			$
			\frac{\mu(f(x))^j}{f_3(x)}
			\big(
			vf_3(x)+uvf_{3,4}(x)
			\big)
			=
			uv\frac{\mu(f(x))^j}{f_3(x)}f_{3,4}(x)
			\in
			\mathcal I.
			$
			Therefore, it follows that 
			$
			\frac{\mu(f(x))^j}{f_3(x)}f_{3,4}(x) 
			\in
			\Tor_{uv}(\mathcal I),
			$ and therefore, $
			f_4(x)\mid_r
			\frac{\mu(f(x))^j}{f_3(x)}f_{3,4}(x).
			$
			
			\item
			Write $	\frac{\mu(f(x))^j}{f_2(x)}
				\big(
				uf_2(x)+vf_{2,3}(x)+uvf_{2,4}(x)
				\big)
					-
				\frac{\mu(f(x))^j}{f_2(x)}
				\frac{f_{2,3}(x)}{f_3(x)}
				\big(
				vf_3(x)+uvf_{3,4}(x)
				\big)\linebreak
				\mathrm{as}~
				uv\frac{\mu(f(x))^j}{f_2(x)}
				\bigg(
				f_{2,4}(x)
				-
				\frac{f_{2,3}(x)}{f_3(x)}f_{3,4}(x)
				\bigg).
		$ Hence, $
			\frac{\mu(f(x))^j}{f_2(x)}
			\bigg(
			f_{2,4}(x)
			-	\frac{f_{2,3}(x)}{f_3(x)}f_{3,4}(x)
			\bigg)
			\in
			\Tor_{uv}(\mathcal I),
			$
			and therefore,
			$
			f_4(x)\mid_r
			\frac{\mu(f(x))^j}{f_2(x)}
			\bigg(
			f_{2,4}(x)
			-
			\frac{f_{2,3}(x)}{f_3(x)}f_{3,4}(x)
			\bigg).
			$
			
			\item
			Since
		$v\big(
				f_1(x)+uf_{1,2}(x)+vf_{1,3}(x)+uvf_{1,4}(x)
				\big)
				-
				\frac{f_1(x)}{f_3(x)}
				\big(
				vf_3(x)+uvf_{3,4}(x)
				\big)\linebreak
		\mathrm{equals}~
				uv
				\bigg(
				f_{1,2}(x)
				-
				\frac{f_1(x)}{f_3(x)}f_{3,4}(x)
				\bigg)
				\in	\mathcal I,$	we obtain
			$f_4(x)\mid_r
			\bigg(
			f_{1,2}(x)
			-
		\frac{f_1(x)}{f_3(x)}f_{3,4}(x)
			\bigg).
			$
			
			\item
			Finally, since
		$
				u\big(
				f_1(x)+uf_{1,2}(x)+vf_{1,3}(x)+uvf_{1,4}(x)
				\big)
				-
				\frac{f_1(x)}{f_2(x)}
				\big(
				uf_2(x)+vf_{2,3}(x)+uvf_{2,4}(x)
				\big)\linebreak
				+
				\frac{f_1(x)f_{2,3}(x)}
				{f_2(x)f_3(x)}
				\big(
				vf_3(x)+uvf_{3,4}(x)
				\big)
				=
				uv
				\bigg(
				f_{1,3}(x)
				-
				\frac{f_1(x)}{f_2(x)}f_{2,4}(x)
				+
				\frac{f_1(x)}
				{f_2(x)f_3(x)}
				f_{2,3}(x)f_{3,4}(x)
				\bigg)\linebreak
				\in
				\mathcal I.
		$ Hence,
			$
			f_4(x)\mid_r
			\bigg(
			f_{1,3}(x)
			-
			\frac{f_1(x)}{f_2(x)}f_{2,4}(x)
			+
			\frac{f_1(x)}
			{f_2(x)f_3(x)}
			f_{2,3}(x)f_{3,4}(x)
			\bigg).
			$
		\end{enumerate}
	\end{proof}

	The following proposition characterizes free skew polycyclic codes of length $jl$ over $R_{u^2,v^2,p^m}$. Its proof follows from standard arguments for free skew cyclic codes \cite{SN23}  and is therefore omitted.
\begin{proposition}\label{free ideal}A left ideal  $\mathcal{I}$ of
	$
	\mathcal{R}^{jl}_{u^2,v^2,f}$
	is~a~free~$R_{u^2,v^2,p^m}$-module if and only if \linebreak$
	\mathcal{I}=\langle g(x)\rangle,
	$
	where $g(x)$ is a monic right divisor of $f(x)^j$ in
	$R_{u^2,v^2,p^m}[x;\Theta]$. If $\deg(g(x))=k$, then
	$
	\{g(x),xg(x),\ldots,x^{jl-k-1}g(x)\}
	$
	is an $R_{u^2,v^2,p^m}$-basis of $\mathcal{I}$, where $l=\deg(f(x))$.
	Consequently,
	$
	\operatorname{rank}(\mathcal{I})=jl-k.
	$
\end{proposition}

Working as in  \cite[Proposition 3.3]{KGP15} we obtain the following corollary using  \Cref{ideals R_{u^2, v^2, p^m}} and  \Cref{free ideal}.
\begin{corollary}

	In the ongoing notation, the left ideal  $\mathcal{I}$ of  $	\mathcal{R}^{jl}_{u^2,v^2,f}$ corresponds to a free skew polycyclic code, if and only if,
	$
	f_1(x)=f_4(x).
	$	In this case,
	$
	\mathcal{I}=
	\mathcal{R}^{jl}_{u^2,v^2,f}(
	f_1(x)+u f_{1,2}(x)+v f_{1,3}(x)+uv f_{1,4}(x)).
	$
	Moreover,
	$
	f_1(x)+u f_{1,2}(x)+v f_{1,3}(x)+uv f_{1,4}(x)
	$
	is a monic right divisor of $f(x)^j$ in
	$R_{u^2,v^2,p^m}[x;\Theta]$.
\end{corollary}

\begin{remark}
	In the special case $m=1$ and when the defining central polynomial satisfies
	$f(x)^j=x^n-1$, the above results recover and extend the corresponding
	results in the skew setting, as given in \cite[Section 3]{KGP15}.
\end{remark}

We now apply the above results to skew $(\lambda,\Theta)$-constacyclic codes of length $np^s$ over $R_{u^2,v^2,p^m}$.
	\subsection{Structural Decomposition Over $R_{u^2,v^2,p^m}$}
Suppose $(n,p)=1$, $s\ge0$, and let $\Theta \in \Aut(	R_{u^2,v^2,p^m})$ be such that the order of $\Theta$ divides $np^s$. 
Assume that $\lambda\in 	R_{u^2,v^2,p^m}^*$ satisfies $\Theta(\lambda)=\lambda$. Then, by \Cref{central polynomials}, the polynomial $x^{np^s}-\lambda$ is central in $	R_{u^2,v^2,p^m}[x;\Theta]$.
As $	R_{u^2,v^2,p^m}$ is finite local ring, by using	\cite[Chapter XX, Exercise 7]{McD74}, there exist irreducible  polynomials $f_1(x),  f_2(x),\ldots, f_r(x)$ in $	R_{u^2,v^2,p^m}[x; \Theta]$ such that
$x^{np^s}-\lambda = f_1(x)f_2(x) \cdots f_r(x).$
Assume that each $f_j(x)$ is also a central polynomial in $	R_{u^2,v^2,p^m}[x;\Theta]$. Then we have $
x^{np^s}-\lambda = f_1(x)^{k_1} f_2(x)^{k_2} \cdots f_t(x)^{k_t},$ where  $f_j(x),  ~1 \le j \le t$,
are  pairwise coprime.
 For each $1\le j\le t$, define
\[
\mathfrak{f}_j(x)=\frac{x^{np^s}-\lambda}{f_j(x)^{k_j}}.
\]
Since $\gcd_{r}(f_j(x)^{k_j},\mathfrak{f}_j(x))=1,$
there exist polynomials $v_j(x),w_j(x)\in 	R_{u^2,v^2,p^m}[x;\Theta]$ such that $v_j(x)\mathfrak{f}_j(x)+w_j(x)f_j(x)^{k_j}=1.$
Set $\varepsilon_j(x)
=
v_j(x)\mathfrak{f}_j(x)
\mod{x^{np^s}-\lambda}.$

Define
\[
\mathcal{R}^{np^s}_{u^2,v^2,\lambda}
=
\frac{R_{u^2,v^2,p^m}[x;\Theta]}
{\langle x^{np^s}-\lambda\rangle} ~\text{and} ~	\mathcal{R}^{np^s}_{\alpha}
=
\frac{\mathbb{F}_{p^m}[x;\theta]}
{\langle x^{np^s}-\alpha\rangle},
\] where $\alpha=\mu(\lambda)$.

Then
\[
\sum_{j=1}^t\varepsilon_j(x)=1,\qquad
\varepsilon_j(x)^2=\varepsilon_j(x)~ \forall ~j
~~\text{and}~~
\varepsilon_j(x)\varepsilon_\ell(x)=0
~~\forall~j,~ j\neq \ell .\]
 Hence $\{\varepsilon_1(x),\ldots,\varepsilon_t(x)\}$	forms a complete set of orthogonal  idempotents in $\mathcal{R}^{np^s}_{u^2,v^2,\lambda}$.\linebreak
Consequently, by the  Chinese Remainder Theorem,
\[
\mathcal{R}^{np^s}_{u^2,v^2,\lambda}
\cong
\bigoplus_{j=1}^t
\mathcal{R}^{l_jk_j}_{u^2,v^2,f_j},
\]
where
\[
\mathcal{R}^{l_jk_j}_{u^2,v^2,f_j}
=
\frac{R_{u^2,v^2,p^m}[x;\Theta]}
{\langle f_j(x)^{k_j}\rangle}, ~k_j=\deg(f_j).
\] 

Moreover, every skew $(\lambda,\Theta)$-constacyclic code
$
\mathcal{C}
\subseteq
\mathcal{R}^{np^s}_{u^2,v^2,\lambda}
$
decomposes uniquely as \linebreak
$
\mathcal{C}
=
\bigoplus_{j=1}^t
\varepsilon_j(x)\mathcal{C}_j,
$
where each $\mathcal C_j$ is a left ideal of
$\mathcal{R}^{l_jk_j}_{u^2,v^2,f_j}$, i.e., a skew
$(f_j(x)^{k_j},\Theta)$-polycyclic code. Therefore, the classification of skew $(\lambda,\Theta)$-constacyclic codes over
$R_{u^2,v^2,p^m}$ reduces to the study of left ideals of
$
\mathcal{R}^{lj}_{u^2,v^2,f}
=
\frac{R_{u^2,v^2,p^m}[x;\Theta]}
{\langle f(x)^j\rangle},
$
where $f(x)^j$ is a power of an irreducible central divisor of
$x^{np^s}-\lambda$ and $\deg(f(x))=l$,  for  $j\in\N$.
This decomposition allows us to describe the left ideals of $\mathcal{R}^{np^s}_{u^2,v^2,\lambda}$  in terms of the left ideals obtained in the previous subsection.

\section{Skew $(\lambda, \Theta)$-constacyclic codes of length $p^s$.}
In this section, we will study skew $(\lambda, \Theta)$-constacyclic codes of length $p^s$ over $R_{u^2, v^2, p^m}$, where $\lambda \in R_{u^2, v^2, p^m}^*$ such that $\Theta(\lambda)=\lambda$. 
Note that the  units of $ R_{u^2,v^2, p^m}$  fall into five types, $\lambda_1=\alpha,
 \lambda_2=\alpha+\gamma v+\delta uv,  \lambda_3=\alpha+\delta_1 uv, \lambda_4= \alpha+\beta u+\delta uv,  \lambda_5= \alpha+\beta u+\gamma v+\delta uv,$
where  $\alpha,  \beta,\gamma,\delta_1\in\mathbb{F}_{p^m}^*$ and  $\delta \in\mathbb{F}_{p^m}$.
Hence, for $1 \leq i \leq 5$ we successively study skew $(\lambda_i, \Theta)$-constacyclic codes of length $p^s$  over $R_{u^2, v^2, p^m}$.
We continue to use the notation
$\mathcal{R}^{p^s}_{u^2,v^2, \lambda_i}\linebreak
=
\frac{R_{u^2, v^2, p^m}[x;\Theta]}
{\langle x^{p^s}-\lambda_i\rangle}, \text{where}~ \lambda_i\in R_{u^2,v^2, p^m}^{*}, 1 \leq i \leq 5 .$
We use the fact that skew $(\lambda_i, \Theta)$-constacyclic codes of length $p^s$  over $R_{u^2, v^2, p^m}$ are left ideals of  $\mathcal{R}^{p^s}_{u^2,v^2, \lambda_i}$ and hence, provide the description of these ideals. 

Denote	the set of all monic proper  divisors of $x^{p^s}-\alpha$  by
$\mathcal{B}^{p^s}_{\alpha}.$
By \Cref*{ideals R_{u^2, v^2, p^m}} every left ideal of $\mathcal{R}^{p^s}_{u^2,v^2, \alpha}$ is of the form
	\begin{align*}
	\mathcal{I} 
	= & \mathcal{R}^{p^s}_{u^2,v^2,\lambda_i}
	\big(f_1(x)+u f_{1,2}(x)+v f_{1,3}(x)+uv f_{1,4}(x)\big) 
	+\\& \mathcal{R}^{p^s}_{u^2,v^2,\lambda_i}
	\big(u f_2(x)+v f_{2,3}(x)+uv f_{2,4}(x)\big) 
	+\\&\mathcal{R}^{p^s}_{u^2,v^2,\lambda_i}
	\big(v f_3(x)+uv f_{3,4}(x)\big) 
	+\\&\mathcal{R}^{p^s}_{u^2,v^2,\lambda_i}
	\big(uv f_4(x)\big),
\end{align*}
	where $f_{i,j}(x)\in
	\mathcal{R}^{p^s}_{\alpha}$, for every $i,j$ and  $f_i(x)$ is either $0$ or $f_i(x)\in   \mathcal{B}^{p^s}_{\alpha}$, for every $i$. In the later case, $
f_4(x)\mid_r f_i(x),
f_i(x)\mid_r f_1(x)$ and
$\deg(f_{i,j}(x))<\deg(f_j(x)).
$
Further, the polynomials $f_{i,j}(x)$ satisfying these conditions are unique.

\subsection{Skew $(\lambda_1, \Theta)$-constacyclic codes, $\lambda_1=\alpha$}
\begin{theorem}\label{alpha classification modified}
	The left ideals of the ring $\mathcal{R}^{p^s}_{u^2, v^2, \alpha}$ are of  the following types :
	\begin{description}
		\item[Type I] Trivial ideals: $0, \mathcal{R}^{p^s}_{u^2, v^2, \alpha}.$
		\item[Type II] Principal left ideals with non-monic generators:
		\begin{description}
			\item[(i)] $\mathcal{I}=\mathcal{R}^{p^s}_{u^2, v^2, \alpha}(uvf_4(x)),$ where $1 \leq \deg(f_4(x)) < p^s$.
			\item[(ii)] $\mathcal{I}=\mathcal{R}^{p^s}_{u^2, v^2, \alpha}(vf_3(x)+uvf_{3,4}(x)),$ where $1 \leq \deg(f_{3,4}(x)) < \deg(f_3(x)) < p^s$.
			\item[(iii)]$\mathcal{I}=\mathcal{R}^{p^s}_{u^2, v^2, \alpha}(uf_2(x)+vf_{2,3}(x)+uvf_{2,4}(x)),$ where $f_2(x) \neq 0$, and  $1 \leq \deg(f_{2,3}(x))
			~\text{or} \linebreak \deg(f_{2,4}(x)) < \deg(f_2(x)) < p^s$.
		\end{description}
		\item[Type III] Principal left ideals with monic generators:
		
	\noindent	$\mathcal{I}=\mathcal{R}^{p^s}_{u^2, v^2, \alpha}(f_1(x)+u f_{1,2}(x)+v f_{1,3}(x)+uv f_{1,4}(x)),$  where 
		$1 \leq \deg(f_{1,j}(x)) < \deg(f_1(x)) < p^s$, for $j>1$ and $f_{1,4}(x) \neq 0$.
		\item[Type IV] Non-principal left ideals:
		
		\begin{description}
			\item[(i)]
			 \(
			\begin{aligned}[t]
			\mathcal{I}=& \mathcal{R}^{p^s}_{u^2, v^2, \alpha}(f_1(x)+u f_{1,2}(x)+v f_{1,3}(x)+uv f_{1,4}(x)) +\\& \mathcal{R}^{p^s}_{u^2, v^2, \alpha}(uf_2(x)+vf_{2,3}(x)+uvf_{2,4}(x)) +\\& \mathcal{R}^{p^s}_{u^2, v^2, \alpha}(vf_3(x)+uvf_{3,4}(x)) +\\& \mathcal{R}^{p^s}_{u^2, v^2, \alpha}(uvf_4(x)),
			\end{aligned}
			\)
			\vspace{.2cm}
			
	where \(f_4(x)\mid_r f_i(x)\),  \(\deg f_{i,j}(x)<\deg f_j(x)\) for all \(i,j\) and 
		 \(f_j(x)\mid_r f_1(x)\),\linebreak \(1\leq \deg f_4(x)<\deg f_j(x)<\deg f_1(x)<p^s\) for \(j=2,3\), \(f_{1,4}(x)\neq 0,  \linebreak1 \leq \deg f_{3,4}(x)<\deg f_3(x)\) , and at least one of $f_{2,3}(x)$ or $ f_{2,4}(x)$ is a non-zero polynomial.
			\vspace{.2cm}
			\item[(ii)]$ \begin{aligned}[t]
				\mathcal{I}=& \mathcal{R}^{p^s}_{u^2, v^2, \alpha}(f_1(x)+u f_{1,2}(x)+v f_{1,3}(x)+uv f_{1,4}(x)) +\\& \mathcal{R}^{p^s}_{u^2, v^2, \alpha}(uf_2(x)+vf_{2,3}(x)+uvf_{2,4}(x)) +\\& \mathcal{R}^{p^s}_{u^2, v^2, \alpha}(vf_3(x)+uvf_{3,4}(x)),
			\end{aligned}$
				\vspace{.2cm}
				
			where 	$\deg(f_{1,j}(x)) < \deg(f_1(x))$, $f_{1,4}(x)\neq 0$,
			 $f_3(x)\mid_r f_1(x)$, $f_2(x)\mid_r f_1(x)$ with $ 1 \leq \deg(f_{3,4}(x)) < \deg(f_3(x)) < \deg(f_2(x)) < \deg(f_1(x)) < p^s, $ and at least one of $f_{2,3}(x)$ or $ f_{2,4}(x)$ is a non-zero polynomial.
				\vspace{.2cm}
			\item[(iii)] $\begin{aligned}[t]
			\mathcal{I}=& \mathcal{R}^{p^s}_{u^2, v^2, \alpha}(f_1(x)+u f_{1,2}(x)+v f_{1,3}(x)+uv f_{1,4}(x)) +\\& \mathcal{R}^{p^s}_{u^2, v^2, \alpha}(uf_2(x)+vf_{2,3}(x)+uvf_{2,4}(x)) + \mathcal{R}^{p^s}_{u^2, v^2, \alpha}(uvf_4(x)),\end{aligned}$
				\vspace{.2cm}
			
			where $f_4(x)\mid_r f_2(x) \mid_r f_1(x)$ with  
			$ 1 \leq \deg(f_4(x)) < \deg(f_2(x)) < \deg(f_1(x)) < p^s ,$
			$\deg(f_{1,j}(x)) < \deg(f_1(x))$, $f_{1,4}(x)\neq 0$, and at least one of $f_{2,3}(x)$ or $ f_{2,4}(x)$ is a non-zero polynomial.
				\vspace{.2cm}
			\item[(iv)]$\begin{aligned}[t]
			\mathcal{I}= &\mathcal{R}^{p^s}_{u^2, v^2, \alpha}(f_1(x)+u f_{1,2}(x)+v f_{1,3}(x)+uv f_{1,4}(x))
			+\\& \mathcal{R}^{p^s}_{u^2, v^2, \alpha}(v f_3(x)+uv f_{3,4}(x))
			+ \mathcal{R}^{p^s}_{u^2, v^2, \alpha}(uv f_4(x)),\end{aligned}$
				\vspace{.2cm}
				
			where $\deg(f_{1,j}(x)) < \deg(f_1(x))$, and $f_{1,4}(x)\neq 0$,
			$f_4(x)\mid_r f_3(x)\mid_r f_1(x)$ with  
			$1 \leq \deg(f_{3,4}(x)) < \deg(f_4(x)) < \deg(f_3(x)) < \deg(f_1(x)) < p^s$.
				\vspace{.2cm}
			\item[(v)]$\begin{aligned}[t] \mathcal{I}=& \mathcal{R}^{p^s}_{u^2, v^2, \alpha}(uf_2(x)+vf_{2,3}(x)+uvf_{2,4}(x)) +\\& \mathcal{R}^{p^s}_{u^2, v^2, \alpha}(vf_3(x)+uvf_{3,4}(x)) +\\& \mathcal{R}^{p^s}_{u^2, v^2, \alpha}
			(uvf_4(x)),\end{aligned}$
			
					\vspace{.2cm}
			 where $1 \leq \deg(f_{3,4}(x)) <  \deg(f_4(x)) < \deg(f_3(x)) < p^s, ~\text{and} ~f_4(x)\mid_r f_2(x)$,\linebreak $f_4(x)\mid_r f_3(x)$ with $\deg(f_4(x)) < \deg(f_2(x)) < p^s$, ~and at least one of $f_{2,3}(x)$ or $ f_{2,4}(x)$ is a non-zero polynomial.
				\vspace{.2cm}
			\item[(vi)] $\begin{aligned}[t]
			\mathcal{I}=& \mathcal{R}^{p^s}_{u^2, v^2, \alpha}(f_1(x)+u f_{1,2}(x)+v f_{1,3}(x)+uv f_{1,4}(x)) +\\& \mathcal{R}^{p^s}_{u^2, v^2, \alpha}(uf_2(x)+vf_{2,3}(x)+uvf_{2,4}(x)),
			\end{aligned}$
				\vspace{.2cm}
				
			where $f_2(x)\mid_r f_1(x)$ with 
			$1 \leq \deg(f_2(x)) < \deg(f_1(x)) < p^s ,$
			 $f_{1,4}(x)\neq 0$, and at least one of $f_{2,3}(x)$ or $f_{2,4}(x)$ is a non-zero polynomial.
				\vspace{.2cm}
			\item[(vii)]  $\begin{aligned}[t]\mathcal{I}=& \mathcal{R}^{p^s}_{u^2, v^2, \alpha}(f_1(x)+u f_{1,2}(x)+v f_{1,3}(x)+uv f_{1,4}(x)) +\\& \mathcal{R}^{p^s}_{u^2, v^2, \alpha}(vf_3(x)+uvf_{3,4}(x)), \end{aligned}$
				\vspace{.2cm}
				
			where $f_3(x)\mid_r f_1(x)$ with 
			$1 \leq \deg(f_{3,4}(x)) < \deg(f_3(x)) < \deg(f_1(x)) < p^s $,
			 $\deg(f_{1,j}(x)) < \deg(f_1(x))$, and $f_{1,4}(x)\neq 0$.
				\vspace{.2cm}
			\item[(viii)] $ \begin{aligned}[t]\mathcal{I}=& \mathcal{R}^{p^s}_{u^2, v^2, \alpha}(f_1(x)+u f_{1,2}(x)+v f_{1,3}(x)+uv f_{1,4}(x)) +\\& \mathcal{R}^{p^s}_{u^2, v^2, \alpha}(uvf_4(x)), \end{aligned} $
				\vspace{.2cm}
				
			where  $f_{1,4}(x)\neq 0$, $f_4(x)\mid_r f_1(x)$ with 
			$1 \leq \deg(f_4(x)) < \deg(f_1(x)) < p^s $, and for $j>1$, $\deg(f_{1,j}(x)) < \deg(f_1(x)).$
				\vspace{.2cm}
			\item[(ix)]  $\begin{aligned}[t]\mathcal{I}=& \mathcal{R}^{p^s}_{u^2, v^2, \alpha}(uf_2(x)+vf_{2,3}(x)+uvf_{2,4}(x)) +\\& \mathcal{R}^{p^s}_{u^2, v^2, \alpha}(vf_3(x)+uvf_{3,4}(x)), \end{aligned}$
				\vspace{.2cm}
				
			where $1 \leq \deg(f_2(x)), \deg(f_3(x)) < p^s$, $1 \leq \deg(f_{3,4}(x)) < \deg(f_3(x))$, and at least one of $f_{2,3}(x)$ or $ f_{2,4}(x)$ is a non-zero polynomial.
				\vspace{.2cm}
			\item[(x)] $ \begin{aligned}[t]\mathcal{I}=& \mathcal{R}^{p^s}_{u^2, v^2, \alpha}(uf_2(x)+vf_{2,3}(x)+uvf_{2,4}(x)) +\\& \mathcal{R}^{p^s}_{u^2, v^2, \alpha}(uvf_4(x)), \end{aligned}$ 
				\vspace{.2cm}
				
			where $f_4(x)\mid_r f_2(x)$ with 
			$1 \leq \deg(f_4(x)) < \deg(f_2(x)) < p^s $,
			and at least one of $f_{2,3}(x)$ or $f_{2,4}(x)$ is a non-zero polynomial.
			
			\item[(xi)] $\mathcal{I}= \mathcal{R}^{p^s}_{u^2, v^2, \alpha}(vf_3(x)+uvf_{3,4}(x)) + \mathcal{R}^{p^s}_{u^2, v^2, \alpha}(uvf_4(x)),$ 
			where $f_4(x)\mid_r f_3(x)$ with 
			$1 \leq \deg(f_{3,4}(x)) <  \deg(f_4(x)) < \deg(f_3(x)) < p^s. $
		
		\end{description}
	\end{description}
\end{theorem}

Recall that $\alpha^{p^m}=\alpha$ for all $\alpha\in\mathbb{F}_{p^m}$. Write,	$s = q m + r, \quad 0 \le r \le m-1,$
for integers $q \ge 0$ and $r$.
Define $\alpha_0^{'} = \alpha^{-p^{m-r}}.$
Then, $\alpha_0^{'p^s}
= \left(\alpha^{-p^{m-r}}\right)^{p^{qm+r}}
= \alpha^{-p^{(q+1)m}}
= \alpha^{-1}.$
Moreover, since $\theta(\alpha)=\alpha$, we have
$
\theta(\alpha_0')=\alpha_0'.
$
\begin{remark} \cite[Lemma 3.6]{CAMK26}
It may be noted that for a given  $\lambda = \alpha^n$, $\alpha \in \mathbb{F}_{p^m}^{*}$, $n \in \mathbb{N}$. If $\Theta(\alpha) = \alpha$, then the map 
\[ \Psi: \frac{R_{u^2, v^2, p^m}[x; \Theta]}{\langle x^{np^s} - 1 \rangle} \to \frac{R_{u^2, v^2, p^m}[x; \Theta]}{\langle x^{np^s} - \lambda \rangle} \]
defined by $\Psi(f(x)) = f(\alpha_0^{'} x)$ is a weight-preserving ring isomorphism.
Consequently, skew $\Theta$-cyclic codes of length $np^s$ over $R_{u^2,v^2,p^m}$ are equivalent to skew $(\lambda,\Theta)$-constacyclic codes of the same length.
\end{remark}

In order to write skew $(\lambda_i,\Theta)$-constacyclic codes of length $p^s$ for $2 \leq i \leq 5$, we shall use the following lemma.
\begin{lemma}\label{invertible}
	Let	$\mathcal{R}_{\alpha}^{p^s}
	=\frac{\mathbb{F}_{p^m}[x;\theta]}
	{\langle x^{p^s}-\alpha\rangle},
	~\alpha \in \mathbb{F}_{p^m}^{*}.
	$
	Then the following statements hold:
	\begin{enumerate}
		\item Let $h(x)$ be a non-zero polynomial in  $\mathcal{R}_{\alpha}^{p^s}$. Then $h(x)$ is left invertible in $\mathcal{R}_{\alpha}^{p^s}$ if and only if
		$
		\operatorname{gcd}_r\bigl(h(x),x^{p^s}-\alpha\bigr)=1.
		$
		
		\item Let $g(x)\in \mathcal{R}_{u^2,v^2,\lambda}^{p^s}$. Then $g(x)$ is left invertible in
		$\mathcal{R}_{u^2,v^2,\lambda}^{p^s}$
		if and only if $\mu(g(x))$ is left invertible in
		$\mathcal{R}_{\alpha}^{p^s}$.
	\end{enumerate}
\end{lemma}
\begin{proof}
The first statement follows from \cite[Lemma 3.4]{HS23}. For the second statement, let \linebreak 
$
g(x)=g_0(x)+ug_1(x)+vg_2(x)+uvg_3(x)
\in \mathcal{R}_{u^2,v^2,\lambda}^{p^s}.
$
Suppose that \(\mu(g(x))=g_0(x)\) is left invertible in
	$\mathcal{R}_{\alpha}^{p^s}$. Then \(g_0(x)\) is invertible, and
$
g_0(x)^{-1}g(x)
=
1+uf_1(x)+vf_2(x)+uvf_3(x),
$
where \(f_i(x)=g_0(x)^{-1}g_i(x)\), \(i=1,2,3\). By direct computations, we obtain that 
$(1+uf_1(x)+vf_2(x)+uvf_3(x))
(1-uf_1(x)-vf_2(x)+uv(2f_1(x)f_2(x)-f_3(x)))
=1.$
Hence \(g_0(x)^{-1}g(x)\) is invertible in $\mathcal{R}_{\alpha}^{p^s}$, implying the invertibility of \(g(x)\). The converse follows immediately by applying the homomorphism \(\mu\).
\end{proof}

\subsection{Skew $(\lambda_2, \Theta)$-constacyclic codes, $\lambda_2=\alpha+\gamma v+\delta uv$ }
As mentioned above, $\lambda_2=\alpha+\gamma v+\delta uv,$
 where $\alpha, \gamma \in \mathbb{F}_{p^m}^*,$ and $\delta \in \mathbb{F}_{p^m}.$ In order to determine left ideals of the ambient ring $\mathcal{R}^{p^s}_{u^2,v^2,\alpha, \gamma, \delta}=\frac{R_{u^2, v^2, p^m}[x, \Theta]}{\langle x^{p^s}-(\alpha+\gamma v+\delta uv)\rangle },$ we need the following lemma.
 \begin{lemma} \label{alpha_gamma_delta_nil}
	For an element $(x-\alpha_0)$ of  $\mathcal{R}^{p^s}_{u^2,v^2 \alpha, \gamma, \delta},$ with 	
	$\alpha_0=\alpha^{p^{m-\alpha_r}}$ such that  $\Theta(\alpha)=\alpha.$ Then 
	\begin{enumerate}
		\item 	$\mathcal{R}^{p^s}_{u^2,v^2,\alpha, \gamma, \delta}(x-\alpha_0)^{p^s}=\mathcal{R}^{p^s}_{u^2,v^2,\alpha, \gamma, \delta}(v).$ In particular, $(x-\alpha_0)$ is a nilpotent element in $\mathcal{R}^{p^s}_{u^2,v^2 \alpha, \gamma, \delta}$ with nilpotency index $2p^s.$
		\item Any polynomial $g(x) \in \mathcal{R}^{p^s}_{u^2, v^2, \alpha, \gamma, \delta}$ can be uniquely expressed as  \[
		g(x)=r_{00}+(x-\alpha_0)\sum_{i=1}^{2p^s-1}r^{'}_{0i}(x-\alpha_0)^{i-1}+u\sum_{i=0}^{2p^s-1}r_{1i}^{'}(x-\alpha_0)^{i},
		\] where $r_{00}, r_{0i}^{'}, r_{1i} \in \mathbb{F}_{p^m}.$
	\item Let
	$
	g(x)=r_{00}
	+(x-\alpha_0)\sum_{i=1}^{2p^s-1}r_{0i}'(x-\alpha_0)^{i-1}
	+u\sum_{i=0}^{2p^s-1}r_{1i}'(x-\alpha_0)^i.
	$
	Then, $g(x)$ is invertible if and only if $r_{00}\neq 0$.
	\end{enumerate}
\end{lemma}
\begin{proof}
	Let $(x-\alpha_0) \in \mathcal{R}^{p^s}_{u^2,v^2,\alpha,\gamma,\delta}$. Since $p \mid \binom{p^s}{i}$ for $1\le i\le p^s-1$ and $\Theta(\alpha)=\alpha$, we have
$(x-\alpha_0)^{p^s}
		=x^{p^s}-\alpha_0^{p^s}
		+\sum_{i=1}^{p^s-1}\binom{p^s}{i}x^i(-\alpha_0)^{p^s-i} 
		=x^{p^s}-\alpha
		=v(\gamma+u\delta).$
	Since $\gamma\neq 0$, the element $\gamma+u\delta$ is a unit. Hence
$	\mathcal{R}^{p^s}_{u^2,v^2,\alpha,\gamma,\delta}(x-\alpha_0)^{p^s}
	=	\mathcal{R}^{p^s}_{u^2,v^2,\alpha,\gamma,\delta}(v).$
	As $v^2=0$, it follows that $(x-\alpha_0)^{2p^s}=0$. Therefore, $(x-\alpha_0)$ is nilpotent with nilpotency index $2p^s$.
		For (2), let
	$g(x)=g_1(x)+ug_2(x)+vg_3(x)+uvg_4(x),$
	where $g_i(x)\in\mathbb{F}_{p^m}[x;\Theta]$ and $\deg(g_i(x))\le p^s-1$. Expanding each $g_i(x)$ in powers of $(x-\alpha_0)$ and using
	$v=(\gamma+u\delta)^{-1}(x-\alpha_0)^{p^s},	$
	we obtain
	$g(x)
	=	r_{00}
	+(x-\alpha_0)\sum_{i=1}^{2p^s-1}r'_{0i}(x-\alpha_0)^{i-1}
	+u\sum_{i=0}^{2p^s-1}r'_{1i}(x-\alpha_0)^i,	$
	where $r_{00},r'_{0i},r'_{1i}\in\mathbb{F}_{p^m}$. 
		Statement (3) follows immediately from \Cref{invertible}.
\end{proof}
By \Cref{invertible} and \Cref*{alpha_gamma_delta_nil}, we obtain a complete classification of the left ideals of
$\mathcal{R}^{p^s}_{u^2,v^2,\alpha,\gamma,\delta}$.
The proof follows the same arguments as those in Theorem 4.2 of \cite{Din10} and  Theorem 3.3 of \cite{DKKY20}; therefore, we omit the details.
\begin{theorem}
Distinct ideals of the ring $\mathcal{R}^{p^s}_{u^2,v^2,\alpha,\gamma,\delta}$ are:
\begin{description}
	\item[Type I] Trivial ideals: $\{0\},$ $\mathcal{R}^{p^s}_{u^2,v^2,\alpha,\gamma,\delta}$.
	\item[Type II] Principal ideal with non-monic polynomial generators:
	$$\mathcal{R}^{p^s}_{u^2,v^2,\alpha,\gamma,\delta}(u(x-\alpha)^i),$$ where $0 \leq i \leq 2p^s-1.$
	\item[Type III] Principal ideal with monic polynomial generators:
	$$\mathcal{R}^{p^s}_{u^2,v^2,\alpha,\gamma,\delta}(x-\alpha)^i+\mathcal{R}^{p^s}_{u^2,v^2,\alpha,\gamma,\delta}(u(x-\alpha)^kl(x)),$$ where $1 \leq i \leq 2p^s-1, ~ 0 \leq k < i,$ and either $l(x)$ is $0$ or $l(x)$ is a unit and can be represented as $l(x)=\sum_{j}r_j(x-\alpha)^j$ with $r_j \in \mathbb{F}_{p^m}$, and $r_0 \neq 0$.
	\item[Type IV] Non-principal ideals:
	$$\mathcal{R}^{p^s}_{u^2,v^2,\alpha,\gamma,\delta}(x-\alpha)^i+\mathcal{R}^{p^s}_{u^2,v^2,\alpha,\gamma,\delta }(u\sum_{j=0}^{\tau-1} r_j(x-\alpha)^j)+\mathcal{R}^{p^s}_{u^2,v^2,\alpha,\gamma,\delta}(u(x-\alpha)^{\tau}),$$ where $1 \leq i \leq 2p^s-1,~ r_j \in \mathbb{F}_{p^m}$ and $\tau < \Gamma,$  $\Gamma$ being the smallest integer such that $u(x-\alpha)^{\Gamma} \in \mathcal{R}^{p^s}_{u^2,v^2,\alpha,\gamma,\delta} (x-\alpha)^i+\mathcal{R}^{p^s}_{u^2,v^2,\alpha,\gamma,\delta}(u\sum_{j=0}^{i-1}r_j(x-\alpha)^j).$ Equivalently, it can be written as
	$$\mathcal{R}^{p^s}_{u^2,v^2,\alpha,\gamma,\delta}(x-\alpha)^i+u(x-\alpha)^kl(x))+\mathcal{R}^{p^s}_{u^2,v^2,\alpha,\gamma,\delta}(u(x-\alpha)^{\tau}),$$ with $l(x)$ same as in Type III, and $\deg(l(x))\leq \tau-k-1.$
\end{description}
\end{theorem}
\begin{remark}\begin{enumerate}
		\item	Let $\Gamma$ be the smallest integer such that $u(x-\alpha)^{\Gamma} \in \mathcal{R}^{p^s}_{u^2,v^2,\alpha,\gamma,\delta}(x-\alpha)^i+u(x-\alpha)^k l(x))$. Then 
		$$\Gamma=\begin{cases} i,~~~~~~~~~~~~~~~~~~~~~~~~~\text{ if}~ l(x)=0,\\
			min\{i, 2p^s-i+k\},~~\text{if}~ l(x) \neq 0. 
		\end{cases}$$
		\item The number of distinct skew $(\alpha+\gamma v+ \delta uv, \Theta)$-constacyclic codes of length $p^s$ over $R_{u^2,v^2,p^m}$ is equal to $$\frac{p^{mp^s}(p^m+3)-4}{(p^m-1)^2}+\frac{4p^{mp^s}-4p^s-5}{p^m-1}+p^{mp^s}.$$
	\end{enumerate}
\end{remark}
The preceding theorem gives a complete description of the left ideals of
$\mathcal{R}^{p^s}_{u^2,v^2,\alpha,\gamma,\delta}$ in terms of the general structure obtained in \Cref{ideals R_{u^2, v^2, p^m}}. We now express these left ideals in terms of powers of the nilpotent element $(x-\alpha_0)$.
By \Cref{alpha_gamma_delta_nil},
$
v=(x-\alpha_0)^{p^s}k,
~k=(\gamma+u\delta)^{-1}\in
\mathcal{R}^{p^s}_{u^2,v^2,\alpha,\gamma,\delta}.
$
Substituting this expression for $v$ into the generators appearing in
\Cref{ideals R_{u^2, v^2, p^m}}, every left ideal $\mathcal I$ of
$\mathcal{R}^{p^s}_{u^2,v^2,\alpha,\gamma,\delta}$ can be written as

\begin{align*}
	\mathcal{I}  = & \mathcal{R}^{p^s}_{u^2, v^2, \alpha, \gamma, \delta}(f_1(x)+u f_{1,2}(x)+(x-\alpha_0)^{p^s}k f_{1,3}(x)+u(x-\alpha_0)^{p^s}k f_{1,4}(x))+\\&
	\mathcal{R}^{p^s}_{u^2, v^2,  \alpha, \gamma, \delta}(uf_2(x)+(x-\alpha_0)^{p^s}k f_{2,3}(x)+u(x-\alpha_0)^{p^s}k f_{3,4}(x))+ \\&
	\mathcal{R}^{p^s}_{u^2, v^2,  \alpha, \gamma, \delta}(x-\alpha_0)^{p^s}kf_3(x)+u(x-\alpha_0)^{p^s}kf_{3,4}(x))+ \\&
	\mathcal{R}^{p^s}_{u^2, v^2,  \alpha, \gamma, \delta}(u(x-\alpha_0)^{p^s}kf_4(x)),
\end{align*}where the polynomials satisfy the conditions of
\Cref{ideals R_{u^2, v^2, p^m}}. 
As the leading coefficient of 	\linebreak$f_1(x) + (x - \alpha_0)^{p^s} k f_{1,3}(x)$
is a unit, by the right division algorithm there exist unique skew polynomials 
\( q_1(x), r_1(x) \in R_{u^2,v^2,p^m}[x;\Theta] \) such that
$
(x - \alpha_0)^{p^s}
= q_1(x)\big(f_1(x) + (x- \alpha_0)^{p^s} k f_{1,3}(x)\big) + r_1(x),
$
where either \( r_1(x)=0 \) or	$\deg(r_1(x)) < \deg\big(f_1(x) + (x - \alpha_0)^{p^s} k f_{1,3}(x)\big)
= \deg(f_1(x)).$
If \( r_1(x) \neq 0 \), then since \( f_1(x) \mid_r (x- \alpha_0)^{p^s} \), it follows that
\( f_1(x) \mid_r r_1(x) \), which contradicts the degree condition. Hence,	$(x - \alpha_0)^{p^s}
= q_1(x)\big(f_1(x) + (x - \alpha_0)^{p^s} k f_{1,3}(x)\big).$

As $u\big(f_1(x) + (x- \alpha_0)^{p^s} k f_{1,3}(x)\big) \in \mathcal{I} ,$ hence 
$u(x-\alpha_0)^{p^s}kf_4(x)=kf_4(x)q_1(x)u\big(f_1(x) +\linebreak (x - \alpha_0)^{p^s} k f_{1,3}(x)\big).$ 
 Consequently, the fourth generator may be omitted. This yields
the following description.
\begin{proposition}\label{ideal alpha gamma delta}
	Every left ideal of $ \mathcal{R}^{p^s}_{u^2, v^2, \alpha, \gamma, \delta}$ is of the form 
	\begin{align*}
	\mathcal{I} = & \mathcal{R}^{p^s}_{u^2, v^2, \alpha, \gamma, \delta}(f_1(x)+u f_{1,2}(x)+(x-\alpha_0)^{p^s}k f_{1,3}(x)+u(x-\alpha_0)^{p^s}k f_{1,4}(x))+\\&
	\mathcal{R}^{p^s}_{u^2, v^2,  \alpha, \gamma, \delta}(uf_2(x)+(x-\alpha_0)^{p^s}k f_{2,3}(x)+u(x-\alpha_0)^{p^s}k f_{2,4}(x))+\\&
	\mathcal{R}^{p^s}_{u^2, v^2,  \alpha, \gamma, \delta}((x-\alpha_0)^{p^s}kf_3(x)+u(x-\alpha_0)^{p^s}kf_{3,4}(x)),
\end{align*}
	where $f_{i,j}(x)\in
\mathcal{R}^{p^s}_{\alpha}$, for every $i,j$ and  $f_i(x)$ is either $0$ or $f_i(x)\in   \mathcal{B}^{p^s}_{\alpha}$, for every $i$. In the later case, $
f_4(x)\mid_r f_i(x),
f_i(x)\mid_r f_1(x)$ and
$\deg(f_{i,j}(x))<\deg(f_j(x)).
$
Further, the polynomials $f_{i,j}(x)$ satisfying these conditions are unique.
\end{proposition}

\subsection{Skew $(\lambda_i, \Theta)$-constacyclic codes, $1 \leq i \leq 3$.}
Proceeding as in  \Cref{alpha_gamma_delta_nil}, one can verify that \(x-\alpha_0\) is a
nilpotent element of \(\mathcal{R}^{p^s}_{u^2,v^2,\lambda_i}\) and  every element of the ambient ring admits a unique expansion in powers of \(x-\alpha_0\). Also, the invertible elements of \(\mathcal{R}^{p^s}_{u^2,v^2,\lambda_i}\) can be
 characterized using  \Cref{invertible}. This enables validating  and extending the results provided in \cite[Sections 4, 5]{DKKY20} for the skew setting.
 Furthermore, using \Cref{ideals R_{u^2, v^2, p^m}}, every left ideal of
 $\mathcal{R}^{p^s}_{u^2,v^2,\lambda_i}$ is obtainable by specializing the generators. Therefore, the complete classification of skew $(\lambda_i, \Theta)$-constacyclic codes is obtainable and the details are omitted to avoid repetitive arguments.

We now illustrate the preceding theory through the following examples.
 \begin{example}
 	Let \(\xi\) be a primitive \(15\)-th root of unity in \(\mathbb{F}_{16}\) satisfying
 	\(\xi^4+\xi+1=0\).
 	Set
 	$
 	\alpha=\xi^{10},~\alpha_0^{-1}=\xi^5,
 	$
 	and consider
 	$
 	\mathcal{R}^{4}_{u^2,v^2,\lambda}
 	=
 	\frac{R_{u^2,v^2,16}[x;\Theta]}{\langle x^4-\lambda\rangle}.
 	$
 	Define the automorphism $\theta$ of $\mathbb{F}_{16}$ by
 	\[
 	\theta(a)=a^4,~ a\in\mathbb{F}_{16},
 	\]
 	which extends to an automorphism $\Theta$ of $R_{u^2,v^2,16}$ by
 	\[
 	\Theta(a+ub+vc+uvd)
 	=
 	\theta(a)+u\theta(b)+v\theta(c)+uv\theta(d),
 	\]
 	for all $a,b,c,d\in\mathbb{F}_{16}$. Then $o(\Theta)=2$.
 	
 	For $\lambda=\alpha=\xi^{10}$, we have
 	\[
 	(\xi^{5}x-1)^4
 	=
 	(\xi^{5}x-\xi^{12})^{2}
 	(\xi^{5}x-\xi^{3})^{2}.
 	\]
 	Let
 	\[
 	b(x)
 	=
 	(\xi^{5}x-\xi^{12})^{2}
 	=
 	\xi^{10}x^{2}+x+\xi^{9}.
 	\]
 	The left ideal  $\mathcal{R}_{\alpha}^4(b(x))$  has generator matrix
 	\[
 	G=
 	\begin{pmatrix}
 		\xi^{9} & 1 & \xi^{10} & 0\\
 		0 & \xi^{6} & 1 & \xi^{10}
 	\end{pmatrix}.
 	\]
 	Hence, it determines an $[4,2,3]$ MDS code over $\mathbb{F}_{16}$.
 	
 	Using the MAGMA algebra system \cite{BCP97} and the Gray map defined in the preliminary section, we obtain the parameters corresponding to Gray images of some skew $(\lambda,\Theta)$-constacyclic codes over $R_{u^2,v^2,16}$ for different values of $\lambda \in R_{u^2,v^2,16}^*$. The corresponding parameters are summarized in Table~\ref{tab:gray16}.
 	
 	\begin{table}[ht]
 		\centering
 		\renewcommand{\arraystretch}{1.25}
 		\begin{tabular}{|c|c|c|}
 			\hline
 			Skew constacyclic constant $\lambda$
 			&
 			$\mathcal{I}$
 			&
 			Parameters of $\Phi(\mathcal{I})$
 			\\
 			\hline
 			$\xi^{10}$
 			&
 			$\mathcal{R}^{4}_{u^2,v^2,\lambda}(uvb(x))$
 			&
 			$[16,2,12]$
 			\\
 			\hline
 			$\xi^{10}+uv$
 			&
 			$\mathcal{R}^{4}_{u^2,v^2,\lambda}(uvb(x))$
 			&
 			$[16,2,12]$
 			\\
 			\hline
 			$\xi^{10}+\xi^{5}v$
 			&
 			$\mathcal{R}^{4}_{u^2,v^2,\lambda}(vb(x))$
 			&
 			$[16,4,6]$
 			\\
 			\hline
 			$\xi^{10}+\xi^{5}u$
 			&
 			$\mathcal{R}^{4}_{u^2,v^2,\lambda}(ub(x))$
 			&
 			$[16,4,6]$
 			\\
 			\hline
 			$\xi^{10}+\xi^{5}u+\xi^{5}v$
 			&
 			$\mathcal{R}^{4}_{u^2,v^2,\lambda}(ub(x))
 			+
 			\mathcal{R}^{4}_{u^2,v^2,\lambda}(vb(x))$
 			&
 			$[16,6,6]$
 			\\
 			\hline
 		\end{tabular}
 		\caption{Parameters of Gray images of skew $(\lambda,\Theta)$-constacyclic codes over $R_{u^2,v^2,16}$.}
 		\label{tab:gray16}
 	\end{table}
 \end{example}

 \begin{example}
 	Consider the ring $R_{u^2,v^2,9}$ and the quotient ring
 	$
 	\mathcal{R}^{6}_{u^2,v^2,2}
 	=
 	\frac{R_{u^2,v^2,9}[x;\Theta]}{\langle x^6-2\rangle}.
 	$
 	Let $\theta$ be the Frobenius automorphism of $\mathbb{F}_{9}$ defined as
 	$
 	\theta(a)=a^3,~a\in\mathbb{F}_{9},
 	$  which extends to an automorphism $\Theta$ of $R_{u^2,v^2,9}$ by
 		\[
 	\Theta(a+ub+vc+uvd)
 	=
 	\theta(a)+\theta(b)u+\theta(c)v+\theta(d)uv,
 	\]
 	for all $a,b,c,d\in\mathbb{F}_{9}$. Clearly, $o(\Theta)=2$.
 	
 	Let $\lambda=2$. Since
 	$
 	x^6-2=(x^2+1)^3,
 	$
 	we choose
 	$
 	b(x)=(x^2+1)^2=x^4+2x^2+1.
 	$
 	The left ideal  $\mathcal{R}_{\alpha}^6(b(x))$  has generator matrix
 		\[
 	G=
 	\begin{pmatrix}
 		1&0&2&0&1&0\\
 		0&1&0&2&0&1
 	\end{pmatrix}.
 	\]
 	Hence, it generates a linear code with parameters $[6,2,3]$.
 	
 	Using the Gray map, we compute the Gray image parameters of some skew $(\lambda,\Theta)$-constacyclic codes over $R_{u^2,v^2,9}$ for different choices of $\lambda\in R_{u^2,v^2,9}^{*}$. These parameters are summarized in Table~\ref{tab:graycomparison}.
 \end{example}
 
 \begin{table}[h]
 	\centering
 	\renewcommand{\arraystretch}{1.25}
 	\begin{tabular}{|c|c|c|}
 		\hline
 		Skew constacyclic constant $\lambda$
 		&
 		$\mathcal{I}$
 		&
 		Parameters of $\Phi(\mathcal{I})$
 		\\
 		\hline
 		$2$
 		&
 		$\mathcal{R}^{6}_{u^2,v^2,\lambda}(uv\,b(x))$
 		&
 		$[24,2,12]$
 		\\
 		\hline
 		$2+uv$
 		&
 		$\mathcal{R}^{6}_{u^2,v^2,\lambda}(uv\,b(x))$
 		&
 		$[24,2,12]$
 		\\
 		\hline
 		$2+v+uv$
 		&
 		$\mathcal{R}^{6}_{u^2,v^2,\lambda}(v\,b(x))$
 		&
 		$[24,4,6]$
 		\\
 		\hline
 		$2+u+uv$
 		&
 		$\mathcal{R}^{6}_{u^2,v^2,\lambda}(u\,b(x))$
 		&
 		$[24,4,6]$
 		\\
 		\hline
 		$2+u+v+uv$
 		&
 		$\mathcal{R}^{6}_{u^2,v^2,\lambda}(u\,b(x))
 		+
 		\mathcal{R}^{6}_{u^2,v^2,\lambda}(v\,b(x))$
 		&
 		$[24,6,6]$
 		\\
 		\hline
 	\end{tabular}
 	\caption{Parameters of Gray images of skew $(\lambda,\Theta)$-constacyclic codes over $R_{u^2,v^2,9}$.}
 	\label{tab:graycomparison}
 \end{table}

\section{Decomposition of $\lambda$-Constacyclic Codes of Length $np^s$ over ${R_{u^2, v^2, p^m}}$}

In this section, we restrict ourselves to the case when $\Theta$ is the identity automorphism, thereby studying constacyclic codes of length $np^s$ over $R_{u^2,v^2,p^m}$. In \cite{CCDF18} and \cite{ZTG18}, $\alpha$-constacyclic codes and $\alpha+u \beta$-constacyclic codes  of length $np^s$ have respectively, being studied over the ring
$\mathbb{F}_{p^m}+u\mathbb{F}_{p^m}$.
We extend these results by studying $(\alpha+u\beta+v\gamma+uv\delta)$-constacyclic codes over the ring $R_{u^2, v^2, p^m}=\mathbb{F}_{p^m}+u\mathbb{F}_{p^m}+v\mathbb{F}_{p^m}+uv\mathbb{F}_{p^m}$.

  Let $\lambda=\alpha+u\beta+v\gamma+uv\delta\in R_{u^2,v^2,p^m}^{*}$ and let $\alpha_0 \in\mathbb{F}_{p^m}^{*}$  be such that 
 $\alpha_0^{p^s}=\alpha$. 
 We now study constacyclic codes, bifurcating by the irreducibility of
  \(x^n-\alpha_0\) over  \(\mathbb{F}_{p^m}[x]\). 

\subsection{$x^n-\alpha_0$ is irreducible in $\mathbb{F}_{p^m}[x]$}

\begin{proposition}\label{invertible two varaible 18}
	Each non-zero polynomial $f(x)$ in $\mathbb{F}_{p^m}[x]$ with degree less than $n$ is invertible in $\mathcal{R}^{np^s}_{u^2,v^2,\lambda}$, i.e., there exists $g(x) \in R_{u^2,v^2,p^m}[x]$ such that $f(x)g(x)\equiv 1 \mod (x^{np^s}-\lambda).$
\end{proposition}

\begin{proof}
	We prove the result by inducting on the degree of the polynomial. If the degree of $f(x)$ is zero, then clearly $f(x)\in \mathbb{F}_{p^m}^{*}$, and hence $f(x)$ is invertible in $\mathcal{R}^{np^s}_{u^2,v^2,\lambda}$.
	
	Now assume that every non-zero polynomial in $\mathbb{F}_{p^m}[x]$ with degree strictly less than $k$ is invertible in $\mathcal{R}^{np^s}_{u^2,v^2,\lambda}$. Let $f(x)\in \mathbb{F}_{p^m}[x]$ be a non-zero polynomial such that $0<\deg(f)=k<n$. By the division algorithm, there exist polynomials $q(x),r(x)\in \mathbb{F}_{p^m}[x]$ such that $x^n-\alpha_0=f(x)q(x)+r(x)$, where either $r(x)=0$ or $\deg(r)<k$. Since $x^n-\alpha_0$ is irreducible over $\mathbb{F}_{p^m}$ and $\deg(f)<n$, it follows that $r(x)\neq 0$. Hence $\deg(r)<k$.
	
	Raising both sides to the $p^s$-th power, we obtain $x^{np^s}-\alpha=f(x)^{p^s}q(x)^{p^s}+r(x)^{p^s}$. Therefore, in the ring $\mathcal{R}^{np^s}_{u^2,v^2,\lambda}$, we have $f(x)^{p^s}q(x)^{p^s}+r(x)^{p^s}-u\beta-v\gamma-uv\delta=0$.
		By the induction hypothesis, $r(x)$ is invertible in $\mathcal{R}^{np^s}_{u^2,v^2,\lambda}$. Put $h(x)=r(x)^{p^s}-u\beta$. Since $u^2=0$, the element $h(x)$ is invertible, with inverse $h(x)^{-1}=r(x)^{-p^s}+u\beta r(x)^{-2p^s}$. Moreover, since $(v\gamma+uv\delta)^2=0$, the element $h(x)-v\gamma-uv\delta$ is invertible, with inverse $(h(x)-v\gamma-uv\delta)^{-1}=h(x)^{-1}+(v\gamma+uv\delta)h(x)^{-2}$.
It follows that $f(x)^{-1}=-f(x)^{p^s-1}q(x)^{p^s}((uv\delta +v\gamma)h(x)^{-2}+h(x)^{-1})$, which gives $f(x)$ is invertible in $\mathcal{R}^{np^s}_{u^2,v^2,\lambda}.$ 

\end{proof}
\begin{lemma}
	Let $\lambda=\alpha+u\beta+v\gamma+uv\delta$
	be a unit of \(R_{u^2,v^2,p^m}\), where
	\(\alpha\in \mathbb{F}_{p^m}^{*}\) and
	\(\beta,\gamma,\delta\in \mathbb{F}_{p^m}\). Let
	\(\alpha_0\in \mathbb{F}_{p^m}^{*}\) be such that
	$
	\alpha_0^{p^s}=\alpha.
	$
	In the residue ring
	$
	\mathcal R_{u^2,v^2,\lambda}^{np^s}$
	we have
	\[
	\left\langle (x^n-\alpha_0)^{p^s}\right\rangle
	=
	\left\langle u\beta+v\gamma+uv\delta\right\rangle.
	\]
	Moreover, \(x^n-\alpha_0\) is nilpotent and its nilpotency index is \(2p^s\)
	if either \(\beta\gamma=0\), or  if \(p=2\). If
	\(\beta\gamma\neq 0\) and \(p\neq 2\), then its nilpotency index is \(3p^s\).
\end{lemma}

\begin{proof}

Since \(\alpha_0^{p^s}=\alpha\), we have
$	(x^n-\alpha_0)^{p^s}
	=
	x^{np^s}-\alpha
	=
	u\beta+v\gamma+uv\delta
	$
	in \(\mathcal R_{u^2,v^2,\lambda}^{np^s}\). Hence
	$\left\langle (x^n-\alpha_0)^{p^s}\right\rangle
	=
	\left\langle u\beta+v\gamma+uv\delta\right\rangle.
	$	Now	$
	\left(u\beta+v\gamma+uv\delta\right)^2
	=
	2\beta\gamma uv.
	$ Therefore, if \(\beta=0\), or \(\gamma=0\), or \(p=2\), then
	$
	\left(u\beta+v\gamma+uv\delta\right)^2=0.
	$
	Thus
	$
(x^n-\alpha_0)^{2p^s}=0,
	$
	and so \(x^n-\alpha_0\) is nilpotent with nilpotency index 
	\(2p^s\). On the other hand, if \(\beta\gamma\neq 0\) and \(p\neq 2\), then
	$
	\left(u\beta+v\gamma+uv\delta\right)^2
	=
	2\beta\gamma uv\neq 0,
	$
	whereas
	$
	\left(u\beta+v\gamma+uv\delta\right)^3=0.
	$
		Therefore, in this case, the nilpotency index of \(x^n-\alpha_0\) is
	\(3p^s\).
\end{proof}

\begin{theorem}
	Assume that $x^n-\alpha_0$ is irreducible in
	$\mathbb{F}_{p^m}[x]$. Then the ring
	$\mathcal R_{u^2,v^2,\lambda}^{np^s}$
	is a finite local non-chain ring.  	More precisely, its unique maximal ideal $	\mathfrak m$ is given as follows:
	\[
	\mathfrak m
	=
	\begin{cases}
	
		\left\langle x^n-\alpha_0,\ u\right\rangle,
		& \text{if } \gamma\neq 0, \\[1mm]
		\left\langle x^n-\alpha_0,\ v\right\rangle,
		& \text{if } \gamma=0,\ \beta\neq 0,\\[1mm]
			\left\langle x^n-\alpha_0,\ u,\ v\right\rangle,
		& \text{if } \beta=\gamma=0.
	\end{cases}
	\]
\end{theorem}

\begin{proof}
	Let
	$
	f(x)=f_1(x)+u f_2(x)+v f_3(x)+uv f_4(x)
	$
	be an arbitrary element of
	$\mathcal R_{u^2,v^2,\lambda}^{np^s}$, where
	$f_i(x)\in\mathbb{F}_{p^m}[x]$. By the division algorithm, there exist
	$q_1(x),r_1(x)\in\mathbb{F}_{p^m}[x]$ such that
	\[
	f_1(x)=q_1(x)(x^n-\alpha_0)+r_1(x),
	\]
	where either $r_1(x)=0$ or $\deg r_1(x)<n$. Hence,
	\[
	f(x)=r_1(x)+u f_2(x)+v f_3(x)+uv f_4(x)
	+(x^n-\alpha_0)q_1(x).
	\]
	If $r_1(x)\neq0$, then by \Cref{invertible two varaible 18},
	$r_1(x)$ is invertible in
	$\mathcal R_{u^2,v^2,\lambda}^{np^s}$. Moreover,
	$
	u f_2(x)+v f_3(x)+uv f_4(x)
	+(x^n-\alpha_0)q_1(x)
	$
	is nilpotent. Therefore $f(x)$ is invertible.
	
	Thus $f(x)$ is non-invertible only when $r_1(x)=0$. In this case,
	$
	f(x)\in \left\langle x^n-\alpha_0,\ u,\ v\right\rangle.
	$
	Conversely, the elements $x^n-\alpha_0$, $u$ and $v$ are nilpotent, and
	hence every element of
	$\left\langle x^n-\alpha_0,\ u,\ v\right\rangle$
	is non-invertible. Therefore the set of all non-units is precisely
	$
	\mathfrak m
	=
	\left\langle x^n-\alpha_0,\ u,\ v\right\rangle.
	$
	Hence $\mathfrak m$ is the unique maximal ideal of
	$\mathcal R_{u^2,v^2,\lambda}^{np^s}$, and the ring is local.
	Now we simplify the generating set of $\mathfrak m$ according to the
	values of $\beta$ and $\gamma$. Since
	$\alpha_0^{p^s}=\alpha$ in
	$\mathcal R_{u^2,v^2,\lambda}^{np^s}$ we have
	$
	(x^n-\alpha_0)^{p^s}
	=
	u\beta+v\gamma+uv\delta.
	$
	If $\beta=\gamma=0$, then
	$
	(x^n-\alpha_0)^{p^s}=uv\delta,
	$
	and neither $u$ nor $v$ can be generated by $x^n-\alpha_0$ alone.
	Thus
	$
	\mathfrak m=
	\left\langle x^n-\alpha_0,\ u,\ v\right\rangle.
	$
		If $\gamma\neq0$, then
	$
	(x^n-\alpha_0)^{p^s}
	=
	u\beta+v(\gamma+u\delta).
	$
	Since $\gamma+u\delta$ is a unit, it follows that
	$
	v=
	\big((x^n-\alpha_0)^{p^s}-u\beta\big)(\gamma+u\delta)^{-1}
	\in
	\left\langle x^n-\alpha_0,\ u\right\rangle.
	$
	Therefore
	$
	\mathfrak m=
	\left\langle x^n-\alpha_0,\ u\right\rangle.
	$
	
		Finally, if $\gamma=0$ and $\beta\neq0$, then
	$
	(x^n-\alpha_0)^{p^s}
	=
	u\beta+uv\delta
	=
	u(\beta+v\delta).
	$
	Since $\beta+v\delta$ is a unit, we get
	$
	u=(x^n-\alpha_0)^{p^s}(\beta+v\delta)^{-1}
	\in
	\left\langle x^n-\alpha_0\right\rangle.
	$
	Hence
	$
	\mathfrak m=
	\left\langle x^n-\alpha_0,\ v\right\rangle.
	$

	Since the maximal ideal is not principal, the ring is not a chain ring.
\end{proof}

\subsection{$x^n-\alpha_0$ is reducible in $\mathbb{F}_{p^m}[x]$}
Assume that $x^n - \alpha_0 = f_1(x) f_2(x)\cdots f_r(x)$ is a factorization into pairwise coprime monic irreducible polynomials in $\mathbb{F}_{p^m}[x]$. Since $\gcd(n,p)=1$, the factorization is square-free, hence
$
x^{np^s} - \alpha = \prod_{j=1}^r f_j(x)^{p^s}.
$

\noindent\textbf{Step 1: Factorization over $\mathbb{F}_{p^m} + u\mathbb{F}_{p^m}$.}
Consider the polynomial $x^{np^s} - \alpha - u\beta,$ and from the above factorization, we have,
$$
x^{np^s} - \alpha - u\beta
= f_1(x)^{p^s} f_2(x)^{p^s}\ldots f_r(x)^{p^s} - u\beta.
$$
Call	$f_2(x)^{p^s}f_3(x)^{p^s} \ldots f_r(x)^{p^s}= F_1(x)^{p^s}.$
Since $f_1(x)^{p^s}$ and $F_1(x)^{p^s}$ are coprime, there exist $v_1(x), w(x) \in \mathbb{F}_{p^m}[x]$ such that $v_1(x) f_1(x)^{p^s} + w(x) F_1(x)^{p^s} = 1.$ 
Thus, we obtain
\[
x^{np^s} - \alpha - u\beta
= \big(f_1(x)^{p^s} - u\beta w(x)\big)\big(F_1(x)^{p^s} - u\beta v_1(x)\big).
\]
By Hensel's lemma (cf. \cite{McD74}), $f_1(x)^{p^s} - u\beta w(x)$ and $F_1(x)^{p^s} - u\beta v_1(x)$
are coprime in $\mathbb{F}_{p^m} + u\mathbb{F}_{p^m}[x]$. Next, consider the factorization of $F_1(x)^{p^s} - u\beta v_1(x).$ Since $f_2(x)^{p^s}$ and $F_2(x)^{p^s},$ where
$F_2(x)=f_3(x) \dots f_r(x)$ are coprime in $\mathbb{F}_{p^m}[x]$, there exist $ v_2(x), w_2(x) \in \mathbb{F}_{p^m}[x]$ such that 
\[
v_2(x) f_2(x)^{p^s} + w_2(x)F_2(x)^{p^s}= 1.
\]

Thus,
$$	F_1(x)^{p^s} - u \beta v_1(x) = (f_2(x)^{p^s} - u \beta v_1(x) w_2(x)) 
  (F_2(x)^{p^s} - u \beta v_1(x) v_2(x)).
$$
\noindent Repeating this,
\[
x^{np^s} - \alpha - u \beta = (f_1(x)^{p^s} - u \beta w_1(x)) (f_2(x)^{p^s} - u \beta v_1(x) w_2(x)) \dots (f_r(x)^{p^s} - u \beta v_1(x) v_2(x) \dots v_{r-1}(x)).
\]

\noindent Let $h_j(x) = f_j(x)^{p^s} + u g_j(x) \quad (1 \le j \le r)$, where  $g_1(x) = -\beta w_1(x), ~ g_j(x) = -\beta v_1(x) \dots \linebreak  v_{j-1}(x) w_j(x), 2 \le j \le r-1,$ $g_r(x) = -\beta v_1(x) \dots v_{r-1}(x).$ Clearly $H_1(x)=h_2(x)h_3(x) \dots h_r(x)$ and $h_1(x)$ are coprime in $\mathbb{F}_{p^m} + u \mathbb{F}_{p^m}$. Hence, there exist
\noindent $ \ s_1(x), t_1(x) \in \mathbb{F}_{p^m} + u\mathbb{F}_{p^m}[x]$ such that  
\[ s_1(x)H_1(x) + t_1(x) h_1(x) = 1 \]

\noindent\textbf{Step 2: Factorization over $R_{u^2, v^2, p^m}$.}
\noindent Using the above factorization
 we have,
\begin{align*}
	x^{np^s} - \alpha - u\beta - v\gamma - uv\delta 
	&= h_1(x) h_2(x) \dots h_r(x) - (v\gamma + uv\delta) \\
	&= \left( h_1(x) - (v\gamma + uv\delta) s_1(x) \right) \left( H_1(x) - (v\gamma + uv\delta) t_1(x) \right)
\end{align*}

\noindent Again by Hensel's lemma,
\noindent $h_1(x) - (v\gamma + uv\delta) s_1(x)$ and $H_1(x) - (v\gamma + uv\delta) t_1(x)$ are coprime in $\mathbb{F}_{p^m} + u\mathbb{F}_{p^m} + v\mathbb{F}_{p^m} + uv\mathbb{F}_{p^m}[x]$.

\vspace{1em}
\noindent Next, consider the factorization of $H_1(x) - (v\gamma + uv\delta) t_1(x).$
 Since $h_2(x)$ and $H_2(x)$, where $H_2(x)=h_3(x) \dots h_r(x),$ are coprime in $\mathbb{F}_{p^m} + u\mathbb{F}_{p^m}[x]$, there exist $ s_2(x), t_2(x) \in \mathbb{F}_{p^m} + u\mathbb{F}_{p^m}[x]$ such that
 $$s_2(x)H_2(x)+t_2(x)h_2(x)=1.$$
Hence,
$H_1(x) - (v\gamma + uv\delta) t_1(x) 
= \left( h_2(x) - (v\gamma + uv\delta) t_1(x) s_2(x) \right) 
\left( H_2(x) - (v\gamma + uv\delta) t_1(x) t_2(x) \right).$

Iterating this process, we get
$$
x^{np^s} - \alpha - u\beta - v\gamma - uv\delta
= \prod_{i=1}^r l_i(x),
$$
where $l_i(x) = h_i(x) - (v\gamma + uv\delta)\left(\prod_{j=1}^{i-1} t_j(x)\right)s_i(x)$, for $1\leq i \leq r$.

\noindent\textbf{Step 3: Idempotent decomposition.}

For each $1 \le j \le r$, define
\[
l'_j(x) = \prod_{\substack{1 \le i \le r \\ i \neq j}} l_i(x).
\]
Since $\gcd(l_j(x), l'_j(x) ) = 1$, there exist $r_j(x), q_j(x) \in R_{u^2, v^2, p^m}[x]$ such that $r_j(x)l_j(x) + q_j(x)l'_j(x)  = 1.$\\
Define $\epsilon_j(x) = q_j(x)l'_j(x)  \mod (x^{np^s} - \lambda).$ Then the elements $\epsilon_j(x)$ satisfy
\[
\sum_{j=1}^r \epsilon_j(x) = 1, \quad 
\epsilon_j(x)^2 = \epsilon_j(x), \quad 
\epsilon_i(x)\epsilon_j(x) = 0 \ \text{for } i \ne j.
\]
Hence, by the Chinese Remainder Theorem,
\[
\frac{R_{u^2, v^2, p^m}[x]}{\langle x^{np^s} - \lambda \rangle}
\cong \bigoplus_{j=1}^r \frac{R_{u^2, v^2, p^m}[x]}{\langle l_j(x) \rangle}.
\]
Thus, the study of $\lambda$-constacyclic codes of length $np^s$ over the ring $R_{u^2, v^2, p^m}$ reduces to the study of ideals of the component rings $\frac{R_{u^2, v^2, p^m}[x]}{\langle l_j(x) \rangle}, \quad 1 \le j \le r,$
where $x^{np^s} - \lambda = \prod_{j=1}^r l_j(x)$ is a factorization into pairwise coprime polynomials in $R_{u^2, v^2, p^m}[x]$.
As in \Cref{ideals R_{u^2, v^2, p^m}}, there is no restriction on $\lambda$ or on the choice of the irreducible divisor $f(x)$ of $x^{np^s} - \lambda$ in the commutative setting. Hence, we obtain the explicit structure of ideals of the quotient ring $\frac{R_{u^2, v^2, p^m}[x]}{\langle l_j(x) \rangle}$.
The following example illustrates this construction and gives the corresponding Gray image parameters.
\begin{example}
		Let $\xi$ be a primitive $7$-th root of unity in $\mathbb{F}_{8}.$ Consider
	$
	\mathcal{R}^{2}_{u^2,v^2,\lambda}
	=
	\frac{R_{u^2,v^2,8}[x;\Theta]}{\langle x^2-\lambda\rangle}.
	$
	
	Let $\lambda=1+u \xi+v\xi^2$ and $\alpha=1.$ Since
$x^2-\alpha=x^2-1=(x-1)(x+1),$ 
 choose
$
f(x)=x+1.
$
	Using the Gray map, we compute the Gray image parameters of $\lambda$-constacyclic codes over $R_{u^2,v^2,8}$. These parameters are summarized in Table~\ref{tab:GrayExample}.
		\begin{table}[ht]
		\centering
		\renewcommand{\arraystretch}{1.2}
		\begin{tabular}{|c|c|c|}
			\hline
			 $\mathcal{I}$
			&
			Parameters of $\Phi(\mathcal{I})$
			&
			MDS
			\\
			\hline
			$\langle f\rangle$
			&
			$[8,6,2]$
			&
			No
			\\
			\hline
			$\langle f+u \xi \rangle$
			&
			$\mathbf{[8,6,3]}$
			&
			Yes
			\\
			\hline
			$\langle f+v \xi^{2}\rangle$
			&
			$\mathbf{[8,6,3]}$
			&
			Yes
			\\
			\hline
			$\langle f+uv \xi^{3}\rangle$
			&
			$[8,6,2]$
			&
			No
			\\
			\hline
			$\langle f+u \xi +v \xi^{2}\rangle$
			&
			$[8,6,2]$
			&
			No
			\\
			\hline
			$\langle f+u \xi+v \xi^{2}+uv \xi^{3}\rangle$
			&
			$[8,6,2]$
			&
			No
			\\
			\hline
			$\langle uf\rangle$
			&
			$[8,3,4]$
			&
			No
			\\
			\hline
			$\langle uf+v \xi f\rangle$
			&
			$[8,2,6]$
			&
			No
			\\
			\hline
			$\langle uf+uv \xi^{2}\rangle$
			&
			$[8,3,4]$
			&
			No
			\\
			\hline
			$\langle uf+v \xi f+uv \xi^{2}\rangle$
			&
			$\mathbf{[8,2,7]}$
			&
			Yes
			\\
			\hline
			$\langle vf\rangle$
			&
			$[8,3,4]$
			&
			No
			\\
			\hline
			$\langle vf+uv \xi^{2}\rangle$
			&
			$[8,3,4]$
			&
			No
			\\
			\hline
			$\langle uvf\rangle$
			&
			$\mathbf{[8,1,8]}$
			&
			Yes
			\\
			\hline
		\end{tabular}
		\caption{Parameters of Gray images of $\lambda$-constacyclic codes over $R_{u^2,v^2,8}$.}
		\label{tab:GrayExample}
	\end{table}
	
	Table~\ref{tab:GrayExample} shows that the Gray images of the ideals
	\[
	\langle f+u\xi\rangle,\qquad
	\langle f+v\xi^{2}\rangle,\qquad
	\langle \xi f+v \xi f+uv \xi^{2}\rangle,\qquad
	\langle uvf\rangle
	\]
	are MDS codes. 

\end{example}

	\bibliographystyle{amsalpha}	\bibliography{SkewConsta}

@book{McD74,
title={Finite rings with identity},
author={McDonald, B.R.},
publisher={Marcel Dekker, New York},
year={1974}
}

@article {HS23,
AUTHOR = {Hesari, R.M. and Samei, K.},
TITLE = {Skew constacyclic codes of lengths {$p^s$} and {$2p^s$} over
{$\Bbb F_{p^m} + u \Bbb F _{p^m} $}},
JOURNAL = {Finite Fields Appl.},
FJOURNAL = {Finite Fields and their Applications},
VOLUME = {91},
YEAR = {2023},
PAGES = {Paper No. 102269, 30},
ISSN = {1071-5797,1090-2465},
MRCLASS = {94B15 (16S36)},

MRREVIEWER = {Xiu\ Sheng\ Liu},
}

@article {Din10,
	AUTHOR = {Dinh, H.Q.},
	TITLE = {Constacyclic codes of length {$p^s$} over {$\Bbb F_{p^m}+u\Bbb
	F_{p^m}$}},
	JOURNAL = {J. Algebra},
	FJOURNAL = {Journal of Algebra},
	VOLUME = {324},
	YEAR = {2010},
	NUMBER = {5},
	PAGES = {940--950},
	ISSN = {0021-8693,1090-266X},
	MRCLASS = {94B15 (13M99)},
}

@article {BGU07,
	AUTHOR = {Boucher, D. and Geiselmann, W. and Ulmer, F.},
	TITLE = {Skew-cyclic codes},
	JOURNAL = {Appl. Algebra Engrg. Comm. Comput.},
	FJOURNAL = {Applicable Algebra in Engineering, Communication and
	Computing},
	VOLUME = {18},
	YEAR = {2007},
	NUMBER = {4},
	PAGES = {379--389},
	ISSN = {0938-1279,1432-0622},
	MRCLASS = {94B25 (16P10)},

}

@misc{BMO26,
				title={Skew polycyclic over finite chain rings associated to trinomials}, 
				author={Bajalan, M. and  Martínez-Moro, E. and  Ou-azzou, H.},
				year={2026},
				eprint={2605.03164},
				archivePrefix={arXiv},
				primaryClass={cs.IT},
				note = "(preprint), arXiv:2605.03164 [math.GR]",
			}

@article {JLU12,
	AUTHOR = {Jitman, S. and Ling, S. and Udomkavanich, P.},
	TITLE = {Skew constacyclic codes over finite chain rings},
	JOURNAL = {Adv. Math. Commun.},
	FJOURNAL = {Advances in Mathematics of Communications},
	VOLUME = {6},
	YEAR = {2012},
	NUMBER = {1},
	PAGES = {39--63},
	ISSN = {1930-5346,1930-5338},
	MRCLASS = {94B15 (94B60)},
MRREVIEWER = {Mohammed\ M.\ Al-Ashker},

	}

@article {ZTG18,
	AUTHOR = {Zhao, W. and Tang, X. and Gu, Z.},
	TITLE = {All {$\alpha+u\beta$}-constacyclic codes of length {$np^s$}
	over {$\Bbb F_{p^m}+u\Bbb F_{p^m}$}},
	JOURNAL = {Finite Fields Appl.},
	FJOURNAL = {Finite Fields and their Applications},
	VOLUME = {50},
	YEAR = {2018},
	PAGES = {1--16},
	ISSN = {1071-5797,1090-2465},
	MRCLASS = {94B15 (11T71 94B05)},

}

@article {CCDF18,
	AUTHOR = {Cao, Y. and Cao, Y. and Dinh, H.Q. and Fu, F.W.
	and Gao, Jian and Sriboonchitta, Songsak},
	TITLE = {Constacyclic codes of length {$np^s$} over {$\Bbb
	F_{p^m}+u\Bbb F_{p^m}$}},
	JOURNAL = {Adv. Math. Commun.},
	FJOURNAL = {Advances in Mathematics of Communications},
	VOLUME = {12},
	YEAR = {2018},
	NUMBER = {2},
	PAGES = {231--262},
    MRCLASS = {94B05 (11T71)},

}

@article{BCP97,
    	AUTHOR={Bosma, W. and Cannon, J. and Playoust, C.},
    	TITLE={The Magma Algebra System {I}: The User Language},
    	JOURNAL={Journal of Symbolic Computation},
    	VOLUME={24},
    	YEAR={1997},
    	NUMBER={3-4},
    	PAGES={235--265},
      }

@article{KGP15,
	AUTHOR = {Kewat, P. K. and Ghosh, B. and  Pattanayak, S. },
	TITLE = {Cyclic codes over the ring $\mathbb{Z}_p[u,v]/\langle u^2, v^2, uv-vu \rangle$},
	JOURNAL = {Finite Fields Appl.},
	FJOURNAL = {Finite Fields and their Applications},
	VOLUME = {34},
	PAGES = {161-175},
	YEAR = {2015},

	
}

@article {RPM26,
	AUTHOR = {Raj, R. and Pathak, S. and Maity, D.},
	TITLE = {Skew constacyclic codes of length {$4p^s$} over {$\Bbb
	{F}_{p^m} + u\Bbb {F}_{p^m}$}},
	JOURNAL = {Comput. Appl. Math.},
	FJOURNAL = {Computational \& Applied Mathematics},
	VOLUME = {45},
	YEAR = {2026},
	NUMBER = {9},
	PAGES = {Paper No. 374},
	ISSN = {2238-3603,1807-0302},
	MRCLASS = {94B15 (16S36 94B60)},
	
}

@article {DKKY20,
	AUTHOR = {Dinh, H.Q. and Kewat, P.K. and Kushwaha, S. and
	Yamaka, W.},
	TITLE = {On constacyclic codes of length {$p^s$} over {$\Bbb
	F_{p^m}[u,v]/\langle u^2,v^2,uv-vu\rangle$}},
	JOURNAL = {Discrete Math.},
	FJOURNAL = {Discrete Mathematics},
	VOLUME = {343},
	YEAR = {2020},
	NUMBER = {8},
	PAGES = {111890, 24},
	ISSN = {0012-365X,1872-681X},
	MRCLASS = {94B15 (94B05)},
	MRREVIEWER = {Latif\ A-M.\ Hanna},
}

@misc{CAMK26,
	title={Skew Constacyclic Codes Of Length {$np^s$} over {$\frac{\mathbb{F}_{p^m}[u]}{\langle u^k \rangle}$}}, 
	author={Chahal, S. and Antil, S. and Maheshwary, S. and Khan, M.},
	year={2026},
	eprint={2605.15925},
	archivePrefix={arXiv},
	primaryClass={cs.IT},
	note = "(preprint), arXiv:2605.15925 [Cs.IT]",
}

@misc{TS26,
		title={Skew Polycyclic Codes Over {$\frac{\mathbb{F}_{p^m}[u]}{\langle u^t \rangle}$}}, 
		author={Tiwari, A. and Sarma, R.},
		year={2026},
		eprint={2605.13020},
		archivePrefix={arXiv},
		primaryClass={cs.IT},
		note = "(preprint), arXiv:2605.13020 [Cs.IT]",
		}

@article {SN23,
			AUTHOR = {Si, X. and Niu, C.},
			TITLE = {On skew cyclic codes over {$M_2(\Bbb F_2)$}},
			JOURNAL = {AIMS Math.},
			FJOURNAL = {AIMS Mathematics},
			VOLUME = {8},
			YEAR = {2023},
			NUMBER = {10},
			PAGES = {24434--24445},
			ISSN = {2473-6988},
			MRCLASS = {94B15 (11T71)},
			
}

@article{YK11,
						AUTHOR = {Yildiz, B. and Karadeniz, S.},
						TITLE = {Cyclic codes over {$\Bbb F_2+u\Bbb F_2+v\Bbb F_2+uv\Bbb F_2$}},
						JOURNAL = {Des. Codes Cryptogr.},
						FJOURNAL = {Designs, Codes and Cryptography. An International Journal},
						VOLUME = {58},
						YEAR = {2011},
						NUMBER = {3},
						PAGES = {221--234},
						ISSN = {0925-1022,1573-7586},
						MRCLASS = {94B15 (11T71)},
						}

@article {DKY12,
							AUTHOR = {Dougherty, S.T. and Karadeniz, S. and Yildiz, B.},
							TITLE = {Cyclic codes over {$R_k$}},
							JOURNAL = {Des. Codes Cryptogr.},
							FJOURNAL = {Designs, Codes and Cryptography. An International Journal},
							VOLUME = {63},
							YEAR = {2012},
							NUMBER = {1},
							PAGES = {113--126},
							ISSN = {0925-1022,1573-7586},
							MRCLASS = {94B15 (94B60)},
							}
			
			\end{document}